\documentclass[sigconf,nonacm,10pt]{acmart}
\usepackage{xcolor}

\newcommand{\short}[1]{}
\newcommand{\extended}[1]{#1}

\setcopyright{none}

\author{Leonidas Fegaras}
\affiliation{%
\institution{University of Texas at Arlington}}
\email{fegaras@cse.uta.edu}
\author{Md Hasanuzzaman Noor}
\affiliation{%
\institution{University of Texas at Arlington}}
\email{mdhasanuzzaman.noor@mavs.uta.edu}

\usepackage{amsmath,amssymb,listings,breakurl,epsfig}
\usepackage[only,llbracket,rrbracket]{stmaryrd}
\usepackage{graphicx}
\usepackage{epstopdf}

\newcommand{\ignore}[1]{}
\newcommand{\skiptext}[1]{\hphantom{#1}}
\newcommand{\gt}{\mbox{$\rangle$}}
\newcommand{\lt}{\mbox{$\langle$}}
\newcommand{\s}[1]{\textsf{#1}}
\newcommand{\app}{\dplus}
\newcommand{\gb}{\Updownarrow}

\newcommand{\sem}[1]{\llbracket #1\rrbracket}
\renewcommand{\bar}[1]{\overline{#1}}
\newcommand{\lm}{\lambda}
\newcommand{\sepc}{\mbox{$\pmb{|}$}}
\newcommand{\from}{\leftarrow}
\newcommand{\compr}[2]{\pmb{\{}\, #1\; \sepc\; #2\,\pmb{\}}}
\newcommand{\ocompr}[2]{\pmb{\{}\, #1\; \sepc\; #2\,\pmb{\}}}
\newcommand{\comprr}[1]{\pmb{\{}\, #1\; \sepc\;}
\newcommand{\bag}[1]{\{\!\!\!\{#1\}\!\!\!\}}
\newcommand{\liste}[1]{[\, #1\,]}

\newcommand{\pluseq}[1]{\,\scalebox{0.85}{$#1$}\!\!=}
\newcommand{\opluseq}{\pluseq{\oplus}}
\newcommand{\req}{\scalebox{0.8}{=}\,}
\newcommand{\some}[1]{\pmb{\{}\,#1\,\pmb{\}}}

\newcommand{\why}[1]{\hspace*{6ex}\mbox{{\em (#1)}}\\}
\newtheorem{theorem}{Theorem}[section]

\newtheorem{definition}{Definition}[section]
\newtheorem{innercustomthm}{Theorem}
\newenvironment{customthm}[1]
  {\renewcommand\theinnercustomthm{#1}\innercustomthm}
  {\endinnercustomthm}
\DeclareRobustCommand\dplus{\mathbin{+\!\!+}}

\lstdefinelanguage{scala}{
  morekeywords={abstract,case,catch,class,def,%
    do,else,extends,false,final,finally,%
    for,if,implicit,import,match,mixin,%
    new,null,object,override,package,%
    private,protected,requires,return,sealed,%
    super,this,throw,trait,true,try,%
    type,until,val,var,while,with,yield},
  otherkeywords={=>,<-,<\%,<:,>:,\#,@},
  sensitive=true,
  basicstyle=\sf,
  morecomment=[l]{//},
  morecomment=[n]{/*}{*/},
  morestring=[b]",
  morestring=[b]',
  morestring=[b]"""
}

\title{Translation of Array-Based Loops to Distributed Data-Parallel Programs}

\newcommand{\myabstract}{
Large volumes of data generated by scientific experiments and
simulations come in the form of arrays, while programs that analyze
these data are frequently expressed in terms of array operations in an
imperative, loop-based language. But, as datasets grow larger, new
frameworks in distributed Big Data analytics have become essential
tools to large-scale scientific computing. Scientists, who are
typically comfortable with numerical analysis tools but are not
familiar with the intricacies of Big Data analytics, must now learn
to convert their loop-based programs to distributed data-parallel
programs. We present a novel framework for translating programs expressed as
array-based loops to distributed data parallel programs that is more
general and efficient than related work.\extended{ Although our translations are
over sparse arrays, we extend our framework to handle packed arrays,
such as tiled matrices, without sacrificing performance.} We
report on a prototype implementation on top of Spark and evaluate the
performance of our system relative to hand-written programs.}

\extended{\begin{abstract}\myabstract\end{abstract}
\begin{document}\maketitle}
\short{\begin{document}\maketitle\begin{abstract}\myabstract\end{abstract}}


\section{Introduction}
\renewcommand{\topfraction}{0.9}	
\renewcommand{\bottomfraction}{0.8}	
\setcounter{topnumber}{2}
\setcounter{bottomnumber}{2}
\setcounter{totalnumber}{2}     
\setcounter{dbltopnumber}{2}    
\renewcommand{\dbltopfraction}{0.9}	
\renewcommand{\textfraction}{0.07}	
\renewcommand{\floatpagefraction}{0.9}	
\renewcommand{\dblfloatpagefraction}{0.9}

Most data used in scientific computing and machine learning come in
the form of arrays, such as vectors, matrices and tensors, while
programs that analyze these data are frequently expressed in terms of
array operations in an imperative, loop-based language.  These loops
are inherently sequential since they iterate over these collections by
accessing their elements randomly, one at a time, using array
indexing.  Current scientific applications must analyze enormous
volumes of array data using complex mathematical data processing
methods.  As datasets grow larger and data analysis computations
become more complex, programs written with array-based loops must now
be rewritten to run on parallel or distributed architectures. Most
scientists though are comfortable with numerical analysis tools, such
as MatLab, and with certain imperative languages, such as FORTRAN and
C, to express their array-based computations using algorithms found in
standard data analysis textbooks, but are not familiar with the
intricacies of parallel and distributed computing.  Because of the
prevalence of array-based programs, a considerable effort has been
made to automatically parallelize these loops. Most automated
parallelization methods in High Performance Computing (HPC) exploit
loop-level parallelism by using multiple threads to access the indexed
data in a loop in parallel. But indexed array values that are updated
in one loop step may be used in the next steps, thus creating
loop-carried dependencies, called recurrences. The presence of such
dependencies complicates the parallelization of a loop. DOALL
parallelization~\cite{doall-kavi} identifies and parallelizes loops
that do not have any recurrences, that is, when statements within a
loop can be executed independently. Although there is a substantial
body of work on automated parallelization on shared-memory
architectures in HPC, there is very little work done on applying these
techniques to the new emerging distributed systems for Big Data
analysis (with the notable exceptions of MOLD~\cite{mold:oopsla14}
and \textsc{Casper}~\cite{casper:sigmod18}).

In recent years, new frameworks in distributed Big Data analytics have
become essential tools for large-scale machine learning and scientific
discoveries. These systems, which are also known as Data-Intensive
Scalable Computing (DISC) systems, have revolutionized our ability to
analyze Big Data.  Unlike HPC systems, which are mainly for
shared-memory architectures, DISC systems are distributed
data-parallel systems on clusters of shared-nothing computers
connected through a high-speed network. One of the earliest DISC
systems is Map-Reduce~\cite{dean:osdi04}, which was introduced by
Google and later became popular as an open-source software with Apache
Hadoop~\cite{hadoop}.\extended{ For each Map-Reduce job, one needs to provide
two functions: a map and a reduce. The map function specifies how to
process a single key-value pair to generate a set of intermediate
key-value pairs, while the reduce function specifies how to combine
all intermediate values associated with the same key. The Map-Reduce
framework uses the map function to process the input key-value pairs
in parallel by partitioning the data across a number of compute nodes
in a cluster. Then, the map results are shuffled across a number of
compute nodes so that values associated with the same key are grouped
and processed by the same compute node.} Recent DISC systems, such as
Apache Spark~\cite{spark} and Apache Flink~\cite{flink}, go beyond
Map-Reduce by maintaining dataset partitions in the memory of the
compute nodes. Essentially, in their core, these systems remain
Map-Reduce systems but they provide rich APIs that implement many
complex operations used in data analysis and support libraries for
graph analysis and machine learning.

The goal of this paper is to design and implement a framework that
translates array-based loops to DISC operations. Not only do these
generated DISC programs have to be semantically equivalent to their
original imperative counterparts, but they must also be nearly as
efficient as programs written by hand by an expert in DISC systems.
If successful, in addition to parallelizing legacy imperative code,
such a translation scheme would offer an alternative and more
conventional way of developing new DISC applications.

DISC systems use data shuffling to exchange data among compute nodes,
which takes place implicitly between the map and reduce stages in Map-Reduce and
during group-bys and joins in Spark and Flink. Essentially, all data
exchanges across compute nodes are done in a controlled way using DISC
operations, which implement data shuffling by distributing data based on
some key, so that data associated with the same key are processed
together by the same compute node. Our goal is to leverage this idea
of data shuffling by collecting the cumulative effects of updates at
each memory location across loop iterations and apply these effects in
bulk to all memory locations using DISC operations.  This idea was
first introduced in MOLD~\cite{mold:oopsla14}, but our goal is to
design a general framework to translate loop-based programs using
compositional rules that transform programs piece-wise, without having
to search for program templates to match (as in
MOLD~\cite{mold:oopsla14}) or having to use a program synthesizer (as
in \textsc{Casper}~\cite{casper:sigmod18}).

Consider, for example, the incremental update
$A[e]\pluseq{+}v$ in a loop, for a sparse vector $A$.  
The cumulative effects of all these updates throughout the loop can be performed in bulk by
grouping the values $v$ across all loop iterations by the array index
$e$ (that is, by the different destination locations) and by summing
up these values for each group.
Then the entire vector $A$ can be replaced with these new values.
For instance, assuming that the values of $C$ were zero before the loop, the following program
\begin{lstlisting}[language=java]
for i = 0, 9 do
    C[A[i].K] += A[i].V
\end{lstlisting}
can be evaluated in bulk by grouping the elements $A[i]$ of the vector
$A$ by $A[i].K$ (the group-by key), and summing up all the values
$A[i].V$ associated with each different group-by key. Then the
resulting key-sum pairs are the new values for the vector $C$.
If the sparse vectors $C$ and $A$ are represented as relational tables
with schemas $(I,V)$ and $(I,K,V)$, respectively, then the new values of $C$ can be
calculated as follows in SQL:
\begin{lstlisting}[language=SQL]
insert into C select A.K as I, sum(A.V) as V
              from A group by A.K
\end{lstlisting}
For example, from $A$ on the left we get $C$ on the right:
\begin{center}
{\small\begin{tabular}{|l|l|}\hline
\multicolumn{1}{|c|}{$A(I,K,V)$} & \multicolumn{1}{|c|}{$C(I,V)$}\\\hline\hline
(3,3,10) & (3,23)\\
(8,5,25) & (5,25)\\
(5,3,13)&\\\hline
\end{tabular}}
\end{center}
These results are consistent with the outcome of the loop, which can be unrolled to the updates
\lstinline$C[3]+=10; C[3]+=13;$ \lstinline$C[5]+=25$.

Instead of SQL, our framework uses monoid comprehensions~\cite{jfp17},
which resemble SQL but have less syntactic sugar and are more concise.
Our framework translates the previous
loop-based program to the following bulk assignment that calculates
all the values of $C$ using a bag comprehension
that returns a bag of index-value pairs:
\[\begin{array}{lcl}
C&:=&\compr{(k,+/v)}{(i,k,v)\from A,\,\textbf{group\ by}\;k}
\end{array}\]
A group-by operation in a comprehension lifts each pattern variable
defined before the group-by (except the group-by keys) from some type
$t$ to a bag of $t$, indicating that each such variable must now
contain all the values associated with the same group-by key
value. Consequently, after we group by $k$, the variable $v$ is lifted
to a bag of values, one bag for each different $k$.  In the
comprehension result, the aggregation $+/v$ sums up all the values in
the bag $v$, thus deriving the new values of $C$ for each index $k$.

A more challenging example, which is used as a running example throughout this paper,
is the product $R$ of two square matrices $M$ and $N$
such that $R_{ij}=\sum_k M_{ik}*N_{kj}$.  It can be expressed as
follows in a loop-based language:
\begin{lstlisting}[language=java]
for i = 0, d-1 do
    for j = 0, d-1 do {
        R[i,j] := 0;
        for k = 0, d-1 do
            R[i,j] += M[i,k]*N[k,j]  }
\end{lstlisting}
A sparse matrix $M$ can be represented as a bag of tuples
$(i,j,v)$ such that $v=M_{ij}$. 
This program too can be translated to a single assignment
that replaces the entire content of the matrix $R$ with a new
content, which is calculated using bulk relational operations.
More specifically, if a sparse matrix is implemented as a
relational table with schema \s{(I,J,V)}, matrix multiplication
between the tables $M$ and $N$ can be expressed as follows in SQL:
\begin{lstlisting}[language=SQL]
select M.I, N.J, sum(M.V*N.V) as V
from M join N on M.J=N.I group by M.I, N.J
\end{lstlisting}
As in the previous example, instead of SQL, our framework uses a comprehension
and translates the loop-based program for matrix multiplication to the following assignment:
\[\begin{array}{lcl}
R&:=&\compr{ (i,j,+/v) }{ (i,k,m) \from M,\, (k',j,n) \from N,\, k=k',\\
&&\skiptext{\comprr{ (i,j,+/v) }}\textbf{let}\;v = m*n,\,\textbf{group\ by}\; (i,j) }
\end{array}\]
Here, the comprehension retrieves the values $M_{ik}\in M$ and $N_{kj}\in N$ as triples
$(i,k,m)$ and $(k',j,n)$ so that $k=k'$, and sets
$v=m*n=M_{ik}*N_{kj}$. 
After we group the values by the matrix indexes
$i$ and $j$, the variable $v$ is lifted to a bag of numerical values
$M_{ik}*N_{kj}$, for all $k$. Hence, the aggregation $+/v$ will sum up
all the values in the bag $v$, deriving $\sum_k M_{ik}*N_{kj}$ for the
$ij$ element of the resulting matrix. If we ignore non-shuffling operations, this comprehension is equivalent
to a join between $M$ and $N$ followed by a reduceByKey
operation in Spark.

\subsection{Highlights of our Approach}

Our framework translates a loop-based program in pieces, in a bottom-up fashion
over the abstract syntax tree (AST) representation of the program, by
translating every AST node to a comprehension.
Matrix indexing is translated as follows:
\[M[i,j] = \compr{ m }{ (I,J,m) \from M,\, I=i,\, J=j }.\]
If $M_{ij}$ exists, it will return the singleton bag $\bag{M_{ij}}$,
otherwise, it will return the empty bag. Since any matrix access that
normally returns a value of $t$ is lifted to a comprehension that
returns a bag of $t$, every term in the loop-based program must be
lifted in the same way. For example, the integer multiplication $A*B$ must be
lifted to the comprehension $\compr{ a*b }{ a \from A,\, b \from B }$
over the two bags $A$ and $B$ (the lifted operands) that
returns a bag (the lifted result).
Consequently, the term $M[i,k]*N[k,j]$\extended{in matrix multiplication}
is translated to:
\[\begin{array}{l}
\compr{a*b}{a\from\compr{m}{ (I,J,m) \from M,\, I=i,\, J=k },\\
\skiptext{\comprr{a*b}}b\from\compr{n}{ (I,J,n) \from N,\, I=k,\, J=j }},
\end{array}\]
which, after unnesting the nested comprehensions and renaming some variables, is normalized to:
\[\begin{array}{l}
\compr{m*n}{(I,J,m) \from M,\, I=i,\, J=k,\\
\skiptext{\comprr{m*n}}(I',J',n) \from N,\, I'=k,\, J'=j },
\end{array}\]
which is equivalent to a join between $M$ and $N$.

Incremental updates, such as $R[i,j] \pluseq{+} M[i,k]*N[k,j]$ in matrix
multiplication, accumulate their values across iterations, hence they
must be considered in conjunction with iterations. Consider 
the following loop, where $f(k)$, $g(k)$, and $h(k)$ are terms that may depend
on $k$:
\begin{lstlisting}[language=java]
for k = 0, 99 do #$M[f(k),g(k)] \pluseq{+} h(k)$#
\end{lstlisting}
Suppose now that there are two values, $k_1$ and $k_2\not= k_1$,
that have the same image under both $f$ and $g$, that is, when
$f(k_1)=f(k_2)$ and $g(k_1)=g(k_2)$. Then, $h(k_1)$ and $h(k_2)$
should be aggregated together. In general, we need to bring together
all values $h(k)$ that have the same values for $f(k)$ and
$g(k)$. That is, we need to group by $f(k)$ and $g(k)$ and sum up all
$h(k)$ in each group. This is accomplished by the
comprehension:
\[\begin{array}{l}
\compr{(\,i,\,j,\,+/v\,)}{k\from\textrm{range}(0,99),\,v\from h(k),\\
\hspace*{6ex}\textbf{let}\;i=f(k),\,\textbf{let}\;j=g(k),\;
\textbf{group\ by}\;(i,j)},
\end{array}\]
where $k\from\textrm{range}(0,99)$ is an iterator that corresponds to
the for-loop and the summation $+/v$ sums up all $h(k)$
that correspond to the same indexes $f(k)$ and $g(k)$.

If we apply this method to $R[i,j] \pluseq{+} M[i,k]*N[k,j]$,
which is embedded in a triple-nested loop, we derive:
\[\begin{array}{l}
\compr{(\,i,\,j,\,+/v\,)}{i\from\textrm{range}(0,d-1),\\
\hspace*{5ex}j\from\textrm{range}(0,d-1),\,k\from\textrm{range}(0,d-1),\\
\hspace*{5ex}v\from M[i,k]*N[k,j],\,\textbf{group\ by}\;(i,j)}.
\end{array}\]
After replacing $M[i,k]*N[k,j]$ and unnesting the nested comprehensions, we get:
\[\begin{array}{l}
\compr{(\,i,\,j,\,+/v\,)}{i\from\textrm{range}(0,d-1),\\
\hspace*{3ex}j\from\textrm{range}(0,d-1),\,k\from\textrm{range}(0,d-1),\\
\hspace*{3ex}(I,J,m) \from M,\, I=i,\, J=k,\,(I',J',n) \from N,\, I'=k,\\
\hspace*{3ex}J'=j,\,\textbf{let}\;v=m*n,\,\textbf{group\ by}\;(i,j)}.
\end{array}\]
Joins between a for-loop and a matrix traversal, such as 
\[i\from\textrm{range}(0,d-1),\,(I,J,m) \from M,\, I=i,\]
can be optimized to a matrix traversal, such as 
\[(i,J,m) \from M,\,\textrm{inRange}(i,0,d-1),\]
where the predicate $\mathrm{inRange}(i,0,d-1)$ returns true if
$0\leq i\leq d-1$. 
Based on this optimization, the previous comprehension becomes:
\[\begin{array}{l}
\compr{(\,i,\,j,\,+/v\,)}{ (i,k,m) \from M,\,\textrm{inRange}(i,0,d-1),\\
\skiptext{\comprr{(\,i,\,j,\,+/v\,)}}(k',j,n) \from N,\,k=k',\\
\skiptext{\comprr{(\,i,\,j,\,+/v\,)}}\textbf{let}\;v = m*n,\,\textbf{group\ by}\; (i,j) },
\end{array}\]
which is the desired translation of matrix multiplication.

We present a novel framework for translating array-based loops to DISC programs
using simple compositional rules that translate these loops piece-wise.
Our framework translates an array-based loop to a semantically
equivalent DISC program as long as this loop satisfies some simple
syntactic restrictions, which are more permissive than the
recurrence restrictions imposed by many current systems and can be statically checked at compile-time.
For a loop to be parallelizable, many systems
require that an array should not be both read and updated
in the same loop.  For example, they reject the update $V[i]:=(V[i-1]+V[i+1])/2$ inside a
loop over $i$ because $V$ is read and updated in the same
loop.  But they also reject incremental updates, such as
$V[i]\pluseq{+}1$, because such an update reads from and writes to the
same vector $V$.  Our framework relaxes these restrictions by
accepting incremental updates of the form $V[e_1]\opluseq e_2$ in a
loop, for some commutative operation $\oplus$ and for some terms $e_1$
and $e_2$ that may contain arbitrary array operations, as long as
there are no other recurrences present.  It translates such an
incremental update to a group-by over $e_1$, followed by a reduction
of the $e_2$ values in each group using the operation
$\oplus$. Operation $\oplus$ is required to be commutative because a
group-by in a DISC system uses data shuffling across the computing
nodes to bring the data that belong to the same group together, which
may not preserve the original order of the data.  Therefore, a non-commutative
reduction may give results that are different from those of the
original loop. We have proved the soundness of our
framework by showing that our translation rules are meaning preserving
for all loop-based programs that satisfy our restrictions\short{
(proved in the extended paper~\cite{diablo-extended})}.  Given that
our translation scheme generates DISC operations, this proof implies
that loop-based programs that satisfy our restrictions are
parallelizable. 
Furthermore, the class of loop-based programs that can be handled
by our framework is equal to the class of programs expressed in our
target language, which consists of comprehensions (i.e., basic SQL),
while-loops, and assignments to variables. Some real-world programs
that contain irregular loops, such as bubble-sort which requires
swapping vector elements, are rejected.

Compared to related work (MOLD~\cite{mold:oopsla14}
and \textsc{Casper}~\cite{casper:sigmod18}): {\bf 1)} Our translation
scheme is complete under the given restrictions as it can translate
correctly any program that does not violate such restrictions, while
the related work is very limited and can work on simple loops
only. For example, neither of the related systems can translate
PageRank or Matrix Factorization. {\bf 2)} Our translator is faster
than related systems by orders of magnitude in some cases, since it
uses compositional transformations without having to search for
templates to apply (as in~\cite{mold:oopsla14}) or use a program
synthesizer to explore the space of valid programs (as
in~\cite{casper:sigmod18}). {\bf 3)} Our translations have been
formally verified\short{ (the soundness proof is given in the extended
version of this paper~\cite{diablo-extended})}, while \textsc{Casper}
needs to call an expensive program validator after each program
synthesis.  Our system, called {\em DIABLO} (a Data-Intensive
Array-Based Loop Optimizer), is implemented on top of
DIQL~\cite{diql:BigData,diql}, which is a query optimization framework
for DISC systems that optimizes SQL-like queries and translates them
to Java byte code at compile-time. Currently, DIABLO has been tested
on Spark~\cite{spark}, Flink~\cite{flink}, and Scala's Parallel
Collections.

\extended{
Although our translations are over sparse arrays, our framework can
easily handle packed arrays, such as tiled matrices, without any
fundamental extension. Essentially, the unpack and pack functions that
convert dense array structures to sparse arrays and vice versa, are
expressed as comprehensions that can be fused with those generated by
our framework, thus producing programs that directly access the packed
structures without converting them to sparse arrays first. This fusion
is hard to achieve in template-based translation systems, such as
MOLD~\cite{mold:oopsla14}, which may require different templates for
different storage structures.
}
The contributions of this paper are summarized as follows:
\begin{itemize}
\item We present a novel framework for translating array-based loops
  to distributed data parallel programs that is more general and
  efficient than related work.
\item We provide simple rules for dependence analysis that detect
  recurrences across loops that cannot be handled by our framework.
\extended{
\item We describe how our framework can be extended to handle packed
  arrays, such as tiled matrices, which can potentially result to a
  better performance.
}
\item We evaluate
  the performance of our system relative to hand-written programs
  on a variety of data analysis and machine learning programs.
\end{itemize}

\extended{
This paper is organized as follows. Section~\ref{framework}
describes our framework in detail. Section~\ref{optimizations}
lists some optimizations on comprehensions that are necessary for
good performance. Section~\ref{pack-unpack} explains how our framework
can be used on densely packed arrays, such as tiled matrices.
Finally, Section~\ref{performance} gives some performance results
for some well-known data analysis programs.
}

\section{Related Work}

Most work on automated parallelization in HPC is focused on
parallelizing loops that contain array scans without recurrences
(DOALL loops) and total reductions
(aggregations)~\cite{fisher:pldi94,jiang:pact18}. As a generalization
of these methods, DOACROSS parallelization~\cite{doall-kavi} separates
the loop computations that have no recurrences from the rest of the
loop and executes them in parallel, while the rest of the loop is
executed sequentially. Other methods that parallelize loops with
recurrences simply handle these loops as DOALL computations but they
perform a run-time dependency analysis to keep track of the dynamic
dependencies, and sequentialize some computations if
necessary~\cite{venkat:sc16}.\extended{ Recently, the work by Farzan
and Nicolet~\cite{farzan:pldi17,farzan:pldi19} describes loop-to-loop
transformations that augment the loop body with extra computations to
facilitate parallelization.
Data parallelism is an effective
technique for high-level parallel programming in which the same
computation is applied to all the elements of a dataset in parallel.}
Most data parallel languages limit their support to flat data
parallelism, which is not well suited to irregular parallel
computations. In flat data-parallel languages, the function applied over the
elements of a dataset in parallel must be itself sequential, while in
nested data-parallel languages this function too can be parallel.
Blelloch and Sabot~\cite{nesl} developed a framework to
support nested data parallelism using flattening, which is a technique
for converting irregular nested computations into regular computations
on flat arrays. These techniques have been extended and implemented in
various systems, such as Proteus~\cite{palmer:95}. DISC-based systems
do not support nested parallelism because it is hard to implement in a
distributed setting.  Spark, for example, does not allow nested RDDs
and will raise a run-time error if the function of an RDD operation
accesses an RDD.  The DIQL and DIABLO translators, on the other hand,
allow nested data parallel computations in any form, by translating
them to flat-parallel DISC operations by flattening comprehensions and
by translating nested comprehensions to DISC
joins~\cite{diql:BigData}.

The closest work to ours is MOLD~\cite{mold:oopsla14}. To the best of
our knowledge, this was the first work to identify the importance of
group-by in parallelizing loops with recurrences in a DISC platform.
Like our work, MOLD can handle complex indirect array accesses simply
using a group-by operation. But, unlike our work, MOLD uses a rewrite
system to identify certain code patterns in a loop and translate them
to DISC operations. This means that such a system is as good as its
rewrite rules and the heuristic search it uses to apply the
rules. Given that the correctness of its translations depends on the
correctness of each rewrite rule, each such rule must be written and
formally validated by an expert. Another similar system is
\textsc{Casper}~\cite{casper:sigmod18}, which translates sequential
Java code into semantically equivalent Map-Reduce programs. It uses a
program synthesizer to search over the space of sequential program
summaries, expressed as IRs. Unlike MOLD, \textsc{Casper} uses a
theorem prover based on Hoare logic to prove that the derived
Map-Reduce programs are equivalent to the original sequential
programs. Our system differs from both MOLD and \textsc{Casper} as it
translates loops directly to parallel programs using simple meaning
preserving transformations, without having to search for rules to
apply. The actual rule-based optimization of our translations is done
at a second stage using a small set of rewrite rules, thus separating
meaning-preserving translation from optimization.

Another related work on automated parallelization for DISC systems is
Map-Reduce program synthesis from input-output
examples~\cite{smith:pldi16}, which is based on recent advances in
example-directed program synthesis.\extended{ One important theorem for
parallelizing sequential scans is the third homomorphism theorem,
which indicates that any homomorphism (ie, a parallelizable
computation) can be derived from two sequential scans; a foldl that
scans the sequence from left to right and a foldr that scans it from
right to left. This theorem has been used to parallelize sequential
programs expressed as folds~\cite{morita:pldi07} by heuristically
synthesizing a foldr from a foldl first. Along these lines is
GRAPE~\cite{fan:sigmod17}, which requires three sequential incremental
programs to derive one parallel graph analysis program, although these
programs can be quite similar.
Lara~\cite{lara} is a declarative domain-specific language for collections
and matrices that allows linear algebra operations on matrices to be mixed with
for-comprehensions for collection processing. This deep embedding of matrix
and collection operations with the host programming language facilitates
better optimization. Although Lara addresses matrix inter-operation optimization,
unlike DIABLO, it does not support imperative loops with random matrix indexing.
}
\short{Another system is GRAPE~\cite{fan:sigmod17}, which requires three sequential incremental
programs to derive one parallel graph analysis program, although these
programs can be quite similar.} Another area related to automated
parallelization for DISC systems is deriving SQL queries from
imperative code~\cite{emani:sigmod16}.  Unlike our work, this work
addresses aggregates, inserts, and appends to lists but does not
address array updates.
Finally, our bulk processing of loop updates resembles the framework
described in~\cite{guravannavar:vldb08}, which rewrites a stored
procedure to accept a batch of bindings, instead of a single binding.
That way, multiple calls to a query under different parameters become
a single call to a modified query that processes all parameters in
bulk.  Unlike our work, which translates imperative loop-based
programs on arrays, this framework modifies existing SQL queries and
updates.

\extended{
Many scientific data generated by scientific experiments and
simulations come in the form of arrays, such as the results from
high-energy physics, cosmology, and climate modeling. Many of these
arrays are stored in scientific file formats that are based on array
structures, such as, CDF (Common Data Format), FITS (Flexible Image
Transport System), GRIB (GRid In Binary), NetCDF (Network Common Data
Format), and various extensions to HDF (Hierarchical Data Format),
such as HDF5 and HDF-EOS (Earth Observing System).
Many array-processing systems use special storage techniques, such as
regular tiling, to achieve better performance on certain array
computations. TileDB~\cite{tiledb} is an array data storage
management system that performs complex analytics on scientific data.
It organizes array elements into ordered collections called fragments,
where each fragment is dense or sparse, and groups contiguous array
elements into data tiles of fixed capacity. Unlike our work, the
focus of TileDB is the I/O optimization of array operations by using
small block updates to update the array stores.
SciDB~\cite{scidb:sigmod10,scidb:ssdbm15} is a large-scale data
management system for scientific analysis based on an array data model
with implicit ordering. The SciDB storage manager decomposes arrays
into a number of equal sized and potentially overlapping chunks, in a
way that allows parallel and pipeline processing of array data. Like
SciDB, ArrayStore~\cite{arraystore:sigmod11} stores arrays into
chunks, which are typically the size of a storage block. One of their
most effective storage method is a two-level chunking strategy with
regular chunks and regular tiles. SystemML~\cite{systemML} is an
array-based declarative language to express large-scale machine
learning algorithms, implemented on top of Hadoop. It supports many
array operations, such as matrix multiplication, and provides
alternative implementations to each of
them. SciHadoop~\cite{scihadoop:sc11} is a Hadoop plugin that allows
scientists to specify logical queries over arrays stored in the NetCDF
file format. Their chunking strategy, which is called the Baseline
partitioning strategy, subdivides the logical input into a set of
partitions (sub-arrays), one for each physical block of the input
file. SciHive~\cite{scihive} is a scalable array-based query system
that enables scientists to process raw array datasets in parallel with
a SQL-like query language. SciHive maps array datasets in NetCDF
files to Hive tables and executes queries via Map-Reduce. Based on the
mapping of array variables to Hive tables, SQL-like queries on arrays
are translated to HiveQL queries on tables and then optimized by the
Hive query optimizer. SciMATE~\cite{scimate} extends the Map-Reduce
API to support the processing of the NetCDF and HDF5 scientific
formats, in addition to flat-files. SciMATE supports various
optimizations specific to scientific applications by selecting a small
number of attributes used by an application and perform data partition
based on these attributes. TensorFlow~\cite{tensorflow} is a dataflow
language for machine learning that supports data parallelism on
multi-core machines and GPUs but has limited support for distributed
computing.  Finally, MLlib~\cite{MLlib:mlr16} is a
machine learning library built on top of Spark and includes algorithms
for fast matrix manipulation based on native (C++ based) linear
algebra libraries. Furthermore, MLlib provides a uniform rigid set of
high-level APIs that consists of several statistical, optimization,
and linear algebra primitives that can be used as building blocks for
data analysis applications.
}

\section{Our Framework}\label{framework}

\begin{figure*}
\framebox{
\hspace*{-3ex}\begin{minipage}[l]{4in}
\[\begin{array}{rcll}
\multicolumn{4}{l}{\mbox{\textbf{Type}:}}\\
t & ::= & v & \mbox{basic type (eg, int, float)}\\
&|&v[t] & \mbox{parametric type (eg, vector)}\\
&|&(t_1,\ldots,t_n) & \mbox{tuple type}\\
&|&\lt\,A_1:t_1,\ldots,A_n:t_n\,\gt & \mbox{record type}\\[1ex]
\multicolumn{4}{l}{\mbox{\textbf{Expression}:}}\\
e & ::= & d & \mbox{a destination (an L-value)}\\
&|&e_1\star e_2 & \mbox{any binary operation $\star$}\\
&|&(e_1,\ldots,e_n) & \mbox{tuple construction}\\
&|&\lt\,A_1\req e_1,\ldots,A_n\req e_n\,\gt & \mbox{record construction}\\
&|&const & \mbox{constant (int, float, \ldots) }
\end{array}\]
\end{minipage}\ 
\begin{minipage}[l]{3in}
\[\begin{array}{rcll}
\multicolumn{4}{l}{\mbox{\textbf{Destination} (L-value):}}\\
d & ::= & v & \mbox{variable}\\
&|&d.A & \mbox{record projection}\\
&|&v[e_1,\ldots,e_n] & \mbox{array indexing}\\[1ex]
\multicolumn{4}{l}{\mbox{\textbf{Statement}:}}\\
s & ::= & d\opluseq e & \mbox{incremental update}\\
&|& d:=e & \mbox{assignment}\\
&|&\mathbf{var}\;v:t=e & \mbox{declaration}\\
&|&\mathbf{for}\;v=e_1,e_2\;\mathbf{do}\;s & \mbox{iteration}\\
&|&\mathbf{for}\;v\;\mathbf{in}\;e\;\mathbf{do}\;s & \mbox{traversal}\\
&|&\mathbf{while}\;(e)\;s & \mbox{loop}\\
&|&\mathbf{if}\;(e)\;s_1\,[\;\mathbf{else}\;s_2\,] & \mbox{conditional}\\
&|&\{\,s_1;\,\ldots;\,s_n\} & \mbox{statement block}
\end{array}\]
\end{minipage}
}
\caption{Syntax of loop-based programs}\label{syntax}
\end{figure*}

\subsection{Syntax of the Loop-Based Language}

The syntax of the loop-based language is given in
Figure~\ref{syntax}. This is a proof-of-concept loop-based language;
many other languages, such as Java or C, can be used instead.  Types
of values include parametric types for various kinds of collections,
such as vectors, matrices, key-value maps, bags, lists, etc. To
simplify our translation rules and examples in this section, we do not
allow nested arrays, such as vectors of vectors.  There are two kinds
of assignments, an incremental update $d\pluseq{\oplus}e$ for some
commutative operation $\oplus$, which is equivalent to the update
$d:=d\oplus e$, and all other assignments $d:=e$. To simplify
translation, variable declarations, $\mathbf{var}\;v:t=e$, cannot
appear inside for-loops. There are two kinds of for-loops that can be
parallelized: a for-loop in which an index variable iterates over a
range of integers, and a for-loop in which a variable iterates over
the elements of a collection, such as the values of an array.
Our current framework generates sequential code from a
while-loop. Furthermore, if a for-loop contains a while-loop in its
body, then this for-loop too becomes sequential and it is
treated as a while-loop.  Finally, a statement block contains a
sequence of statements.

\subsection{Restrictions for Parallelization}\label{recurrences}

Our framework can translate for-loops to equivalent DISC programs
when these loops satisfy certain restrictions
described in this section. 
In Appendix~\ref{proofs}, we provide a proof that,
under these restrictions, our transformation rules to be
presented in Section~\ref{translation} are meaning preserving,
that is, the programs generated by our translator are equivalent to
the original loop-based programs. In other words, since our target
language is translated to DISC operations, the loop-based programs that satisfy our
restrictions are parallelizable.

Our restrictions use the following definitions.
For any statement $s$ in a loop-based program, we define the following
three sets of L-values (destinations): the readers $\mathcal{R}\sem{s}$,
the writers $\mathcal{W}\sem{s}$, and the aggregators
$\mathcal{A}\sem{s}$. The {\bf readers} are the L-values read in $s$, the
{\bf writers} are the L-values written (but not incremented) in $s$, and
the {\bf aggregators} are the L-values incremented in $s$.
For example, for the following statement:
\begin{align*}
V[W[i]] \pluseq{+} n*C[i]*C[i+1],
\end{align*}
where $i$ is a loop index,
the aggregators are $\mathcal{A}\sem{s}=\bag{V[W[i]]}$, the readers are
$\mathcal{R}\sem{s}=\bag{W[i],n,C[i],C[i+1]}$, and the writers are $\mathcal{W}\sem{s}=\emptyset$.
Two L-values $d_1$ and $d_2$ {\bf overlap}, denoted by $\mathrm{overlap}(d_1,$ $d_2)$,
if they are the same variable, or they are equal to the projections
$d_1'.A$ and $d_2'.A$ with $\mathrm{overlap}(d_1',d_2')$,
or they are array accesses over the same array name.
The {\bf context} of a statement $s$, $\mathrm{context}(s)$, is the set of outer loop
indexes for all loops that enclose $s$.
Note that, each for-loop must have a distinct loop index variable;
if not, the duplicate loop index is replaced with a fresh variable.
For an L-value $d$, $\mathrm{indexes}(d)$
is the set of loop indexes used in $d$.

An {\bf affine} expression~\cite{aho:book} takes
the form
\[c_0+c_1*i_1+\cdots+c_k*i_k,\]
where $i_1,\ldots,i_k$ are loop indexes and $c_0,\ldots,c_k$ are constants.  For an
L-value $d$ in a statement $s$, $\mathrm{affine}(d,s)$ is true if $d$ is a variable, or
a projection $d'.A$ with $\mathrm{affine}(d',s)$, or an
array indexing $v[e_1,\ldots,e_n]$, where each index $e_i$ is an
affine expression and all loop indexes in $\mathrm{context}(s)$ are
used in $d$. In other words, if $\mathrm{affine}(d,s)$ is true, then
$d$ is stored at different locations for different values of the loop
indexes in $\mathrm{context}(s)$.

\begin{definition}[Affine For-Loop]\label{PFOR-def}
A for-loop statement $s$ is affine 
if $s$ satisfies the following properties:
\begin{enumerate}
\item for any update $d:=e$ in $s$, $\mathrm{affine}(d,s)$;
\item there are no dependencies between any two statements $s_1$ and $s_2$ in $s$,
that is, if there are no L-values\linebreak $d_1\in(\mathcal{A}\sem{s_1}\cup\mathcal{W}\sem{s_1})$ and
$d_2\in\mathcal{R}\sem{s_2}$ such that\linebreak
$\mathrm{overlap}(d_1,d_2)$,
with the following exceptions:
\begin{enumerate}
\item if $d_1\in\mathcal{W}\sem{s_1}$, $d_1=d_2$, and $s_1$ precedes $s_2$;
\item if $d_1\in\mathcal{A}\sem{s_1}$, $d_1=d_2$, $s_1$ precedes $s_2$,
$\mathrm{affine}(d_2,s_2)$,
and $\mathrm{context}(s_1)\cap\mathrm{context}(s_2)=\mathrm{indexes}(d_1)$.
\end{enumerate}
\end{enumerate}
\end{definition}
Restriction~1 indicates that the destination of any non-incre-mental
update must be a different location at each loop iteration. If the update
destination is an array access, the array indexes must be affine and
completely cover all surrounding loop indexes. This restriction does
not hold for incremental updates, which allow arbitrary array indexes
in a destination as long as the array is not read in the same loop.
Restriction~2 combined with exception~(a)
rejects any read and write on the
same array in a loop except when the read is after the write and
the read and write are at the same location ($d_1=d_2$),
which, based on Restriction~1, is a different location at each loop iteration.
Exception~(b) indicates that
if we first increment and then read the same location, 
then these two operations must not be inside a for-loop whose
loop index is not used in the destination.  This is because the
increment of the destination is done within the for-loops whose loop
indexes are used in the destination and across the rest of the
surrounding for-loops. For example, the following loop:
\[\textbf{for}\;i=\ldots\;\textbf{do}\;\{\;\textbf{for}\;j=\ldots\;\textbf{do}\;V[i]\pluseq{+}1;\;W[i]:=V[i]\;\},\]
increments and reads $V[i]$.
The contexts of the first and second
updates are $\bag{i,j}$ and $\bag{i}$, respectively, and
their intersection gives $\bag{i}$, which is equal to the indexes of
$V[i]$.  If there were another statement $M[i,j]:=V[i]$ inside the
inner loop, this would violate Exception~(b) since their context
intersection would have been $\bag{i,j}$, which is not equal to the
indexes of $V[i]$.

An affine for-loop satisfies the following theorem,
\extended{which is proved in Appendix~\ref{proofs}.}
\short{which is proved in the extended version of this paper~\cite{diablo-extended}.}
It is used as the basis of our program translations.
\begin{theorem}\label{for-distribution}
An affine for-loop satisfies:
\begin{align}
&\textbf{for}\;i=\ldots\;\textbf{do}\;\{\,s_1;\,s_2\,\}\nonumber\\
&\hspace*{3ex}\;=\;\{\,\textbf{for}\;i=\ldots\;\textbf{do}\;s_1;\;\textbf{for}\;i=\ldots\;\textbf{do}\;s_2\;\}.\label{for-distr}
\end{align}
\end{theorem}
In fact, our restrictions in Definition~\ref{PFOR-def} were
designed in such a way that all affine for-loops satisfy
this theorem and at the same time are inclusive enough to accept
as many common loop-based programs as possible.
In Appendix~\ref{proofs}, we prove that
our program translations, to be described in Section~\ref{translation},
under the restrictions in Definition~\ref{PFOR-def} are meaning preserving,
which implies that all affine for-loops are parallelizable since the
target of our translations is DISC operations.

For example, the incremental update:
\[\textbf{for}\;i=\ldots\;\textbf{do}\;C[V[i].K]\pluseq{+}V[i].D,\]
which counts all $V[i].D$ in groups that have the same key $V[i].K$,
satisfies our restrictions since it increments but does
not read $C$.  On the other hand, some
non-incremental updates may outright be rejected. For example, the
loop:
\[\textbf{for}\;i=\ldots\;\textbf{do}\;V[i]:=(V[i-1]+V[i+1])/2\]
will be rejected by Restriction~2 because $V$ is both a reader and a writer.
To alleviate this problem, one may rewrite this loop as follows:
\[\begin{array}{l}
\textbf{for}\;i=\ldots\;\textbf{do}\;V'[i]:=V[i];\\
\textbf{for}\;i=\ldots\;\textbf{do}\;V[i]:=(V'[i-1]+V'[i+1])/2,
\end{array}\]
which first stores $V$ to $V'$ and then reads $V'$ to compute $V$. This
program satisfies our restrictions but is not equivalent to
the original program because it uses the previous values of $V$ to
compute the new ones.
Another example is:
\[\textbf{for}\;i=\ldots\;\textbf{do}\;\{\,n:=V[i];\,W[i]:=f(n)\,\},\]
which is also rejected because $n$ is not affine as it does not cover
the loop indexes (namely, $i$). To fix this problem, one may redefine
$n$ as a vector and rewrite the loop as:
\[\textbf{for}\;i=\ldots\;\textbf{do}\;\{\,n[i]:=V[i];\,W[i]:=f(n[i])\,\}.\]
Redefining variables by adding to them more array dimensions is
currently done manually by a programmer, but we believe that
it can be automated when a variable that violates our restrictions is detected.

A more complex example is matrix factorization using gradient
descent~\cite{koren:comp09}. The goal of matrix factorization is to
split a matrix $R$ of dimension $n\times m$ into two low-rank matrices
$P$ and $Q$ of dimensions $n\times l$ and $l\times m$, for small $l$,
such that the error between the predicted and the original matrix
$R-P\times Q$ is below some threshold. One step of matrix
factorization that computes the new values $P$ and $Q$ from the
previous values $P'$ and $Q'$ can be implemented using the following
loop-based program:
\begin{lstlisting}[language=java]
for i = 0, n-1 do
   for j = 0, m-1 do {
      pq := 0.0;
      for k = 0, l-1 do
         pq += P#'#[i,k]*Q#'#[k,j];
      error := R[i,j]-pq;
      for k = 0, l-1 do {
         P[i,k] += a*(2*error*Q#'#[k,j]-b*P#'#[i,k]);
         Q[k,j] += a*(2*error*P#'#[i,k]-b*Q#'#[k,j]); }}
\end{lstlisting}
where \s{a} is the learning rate and \s{b} is the normalization factor
used in avoiding overfitting. This program first computes \s{pq},
which is the $i,j$ element of $P'\times Q'$, and \s{error}, which is
the $i,j$ element of $R-P'\times Q'$. Then, it uses \s{error} to
improve $P$ and $Q$. This program is rejected because the destinations
of the assignments \s{pq := 0.0} and \s{error := R[i,j]-pq}
do not cover all loop indexes, and the read of \s{pq} violates exception~(b)
(since the intersection of the contexts of \s{pq += P'[i,k]*Q'[k,j]}
and \s{error := R[i,j]-pq} is \s{\{i,j\}}, which is not equal to the indexes of \s{pq}).
To rectify these problems, we can convert the variables
\s{pq} and \s{error} to matrices,
so that, instead of \s{pq} and \s{error}, we use \s{pq[i,j]}
and \s{error[i,j]}.

\subsection{Monoid Comprehensions}\label{comprehensions}

The target of our translations consists of monoid comprehensions,
which are equivalent to the SQL select-from-where-group-by-having syntax.
Monoid comprehensions were first introduced and used in the
90's as a formal basis for ODMG OQL~\cite{tods00}.  They were
recently used as the formal calculus for the DISC query languages
MRQL~\cite{jfp17} and DIQL~\cite{diql:BigData}.  The formal semantics
of monoid comprehensions, the query optimization framework, and the
translation of comprehensions to a DISC algebra, are given in our earlier
work~\cite{jfp17,diql:BigData}. Here, we describe the syntax only.

A monoid comprehension has the following syntax:
\[\compr{e}{q_1,\;\ldots,\;q_n},\]
where the expression $e$ is the comprehension head and a qualifier
$q_i$ is defined as follows:
\[\begin{array}{rcll}
\multicolumn{4}{l}{\mbox{\textbf{Qualifier}:}}\\
\hspace*{5ex}q&::=& p\from e & \mbox{generator}\\
&|&\textbf{let}\;p=e & \mbox{let-binding}\\
&|&e & \mbox{condition}\\
&|&\textbf{group\ by}\,p\, [\,: e\,]&\mbox{group-by}\\
\multicolumn{4}{l}{\mbox{\textbf{Pattern}:}}\\
\hspace*{5ex}p& ::=& v &\mbox{pattern variable}\\
&|&(p_1,\ldots,p_n)&\mbox{tuple pattern.}
\end{array}\]
The domain $e$ of a generator $p\from e$ must be a bag.
This generator draws elements from this
bag and, each time, it binds the pattern $p$ to an element. A
condition qualifier $e$ is an expression of type boolean. It is used
for filtering out elements drawn by the generators.
A let-binding $\textbf{let}\;p=e$ binds the pattern $p$ to the result of
$e$. A group-by qualifier uses a pattern $p$ and an optional
expression $e$. If $e$ is missing, it is taken to be $p$. The group-by
operation groups all the pattern variables in the same comprehension
that are defined before the group-by (except the variables in $p$) by
the value of $e$ (the group-by key), so that all variable bindings
that result to the same key value are grouped together. After the
group-by, $p$ is bound to a group-by key and each one of these pattern
variables is lifted to a bag of values.
The result of a comprehension $\compr{e}{q_1,\;\ldots,\;q_n}$ is a bag that contains
all values of $e$ derived from the variable bindings in the qualifiers.

\short{
Comprehensions can be translated to algebraic operations that resemble
the bulk operations supported by many DISC systems, such as groupBy,
join, map, and flatMap. In an earlier work~\cite{jfp17}, we have presented
a general method for identifying
all possible equi-joins in comprehensions, including joins across
deeply nested comprehensions, and translating them to joins and
coGroups.
}

\extended{
Comprehensions can be translated to algebraic operations that resemble
the bulk operations supported by many DISC systems, such as groupBy,
join, map, and flatMap. We use $\bar{q}$ to
represent the sequence of qualifiers $q_1,\;\ldots,\;q_n$, for $n\geq 0$.  
To translate a comprehension $\compr{e}{\bar{q}}$ to the algebra,
the group-by qualifiers are first translated to groupBy operations
from left to right.  Given a bag $X$ of type $\bag{(K,V)}$,
$\mathrm{groupBy}(X)$ groups the elements of $X$ by their first
component of type $K$ (the group-by key) and returns a bag of type
$\bag{(K,\bag{V})}$.  Let $v_1,\ldots,v_n$ be the pattern variables in
the sequence of qualifiers $\bar{q_1}$ that do not appear in the
group-by pattern $p$, then we have:
\begin{align*}
&\compr{e'}{\bar{q_1},\,\textbf{group\ by}\,p: e,\,\bar{q_2}}\\
&=\;\compr{e'}{(p,s)\from\mathrm{groupBy}(\compr{(e,(v_1,\dots,v_n))}{\bar{q_1}}),\\
&\skiptext{=\;\comprr{e'}}\forall i:\;\textbf{let}\;v_i=\compr{v_i}{(v_1,\ldots,v_n)\from s},\;\bar{q_2}}.
\end{align*}
That is, for each pattern variable $v_i$, this rule embeds a
let-binding so that this variable is lifted to a bag that contains all
$v_i$ values in the current group.  Then, comprehensions without any
group-by are translated to the algebra by translating the qualifiers
from left to right:
\begin{align*}
\compr{e'}{p\from e,\,\bar{q}} & =\; \mathrm{flatMap}(\lm p.\,\compr{e'}{\bar{q}},\,e)\\
\compr{e'}{\textbf{let}\;p=e,\,\bar{q}} & =\; \textbf{let}\;p=e\;\textbf{in}\;\compr{e'}{\bar{q}}\\
\compr{e'}{e,\,\bar{q}} & =\; \textbf{if}\;e\;\textbf{then}\;\compr{e'}{\bar{q}}\;\textbf{else}\;\emptyset\\
\compr{e'}{} & =\; \bag{e'}.
\end{align*}
Given a function $f$ that maps an element of type $T$ to a bag of type
$\bag{S}$ and a bag $X$ of type $\bag{T}$, the operation
$\mathrm{flatMap}(f,X)$ maps the bag $X$ to a bag of type $\bag{S}$ by
applying the function $f$ to each element of $X$ and unioning together
the results.  Although this translation generates nested flatMaps from
join-like comprehensions, there is a general method for identifying
all possible equi-joins from nested flatMaps, including joins across
deeply nested comprehensions, and translating them to joins and
coGroups~\cite{jfp17}.

Finally, nested comprehensions can be unnested by the following rule:
\begin{align}
&\compr{e_1}{\bar{q_1},\,p\from\compr{e_2}{\bar{q_3}},\,\bar{q_2}}\nonumber\\
&\hspace*{10ex}=\; \compr{e_1}{\bar{q_1},\,\bar{q_3},\,\textbf{let}\;p=e_2,\,\bar{q_2}}\label{nested-comps}
\end{align}
for any sequence of qualifiers $\bar{q_1}$, $\bar{q_2}$, and
$\bar{q_3}$. This rule can only apply if there is no group-by
qualifier in $\bar{q_3}$ or when $\bar{q_1}$ is empty. It may require
renaming the variables in $\compr{e_2}{\bar{q_3}}$ to prevent variable
capture.
}

\subsection{Array Representation}\label{array-rep}

In our framework, a sparse array, such as a sparse vector or a matrix,
is represented as a key-value map (also known as an indexed set), which is a bag of
type $\bag{(K,T)}$, where $K$ is the array index type and $T$
is the array value type.  More specifically, a sparse vector of type
$\mathrm{vector}[T]$ is captured as a key-value map of type
$\bag{(\mathrm{long},T)}$, while a sparse matrix of type
$\mathrm{matrix}[T]$ is captured as a key-value map of type
$\bag{((\mathrm{long},\mathrm{long}),T)}$.

Merging two compatible arrays is done with the array merging
operation $\lhd$, defined as follows:
\begin{align*}
X\lhd Y\; =\; & \compr{(k,b)}{(k,a)\from X,\,(k',b)\in Y,\,k=k'}\\
& \uplus\;\compr{(k,a)}{(k,a)\from X,\,k\not\in\Pi_1(Y)}\\
& \uplus\;\compr{(k,b)}{(k,b)\from Y,\,k\not\in\Pi_1(X)},
\end{align*}
where $\Pi_1(X)$ returns the keys of $X$.
That is, $X\lhd Y$ is the union of $X$ and $Y$, except when there is
$(k,x)\in X$ and $(k,y)\in Y$, in which case it chooses the latter value, $(k,y)$.
For example, $\bag{(3,10),(1,20)}\lhd\bag{(1,30),(4,40)}$ is equal to
$\bag{(3,10),(1,30),(4,40)}$. 
On Spark, the $\lhd$ operation can be implemented as a coGroup.

An update to a vector $V[e_1]:=e_2$ is equivalent to the assignment
$V:=V\lhd\bag{(e_1,e_2)}$.  That is, the new value of $V$ is the current vector $V$ but with the
value associated with the index $e_1$ (if any) replaced with $e_2$.
Similarly, an update to a matrix $M[e_1,e_2]:=e_3$ is equivalent to the assignment
$M:=M\lhd\bag{((e_1,e_2),e_3)}$.

Array indexing though is a little bit more complex because the indexed element
may not exist in the sparse array. Instead of a value of type $T$,
indexing over an array of $T$ should return a bag of type
$\bag{T}$, which can be $\bag{v}$ for some value $v$ of type
$T$, if the value exists, or $\emptyset$, if the value does not exist.
Then, the vector indexing $V[e]$ is $\ocompr{v}{(i,v)\from V,\,i=e}$,
which returns a bag of type $\bag{T}$. Similarly, the matrix
indexing $M[e_1,e_2]$ is $\ocompr{v}{((i,j),v)\from
M,\,i=e_1,\,j=e_2}$.

We are now ready to express any assignment that involves vectors and
matrices.  For example, consider the matrices $R$, $M$, and $N$ of
type matrix[float]. The assignment:
\begin{align}
&R[i,j] := M[i,k]*N[k,j]\label{mat-simple}
\end{align}
is translated to the assignment:
\begin{align}
&R := R\lhd\compr{((i,j),m*n)}{((i,k),m)\from M,\label{assign:example}\\
&\skiptext{R := R\lhd\comprr{((i,j),m*n)}}((k',j),n)\from N,\,k=k'},\nonumber
\end{align}
which uses a bag comprehension equivalent to a join between the
matrices $M$ and $N$.  This assignment can be derived from
assignment~(\ref{mat-simple}) using simple transformations.  To
understand these transformations, consider the product $X*Y$.  Since
both $X$ and $Y$ have been lifted to bags, because they may contain array
accesses, this product must also be lifted to a comprehension
that extracts the values of $X$ and $Y$, if any, and returns their
product:
\[X*Y = \ocompr{x*y}{x\from X,\,y\from Y}.\]
Given that matrix accesses are expressed as:
\[\begin{array}{l}
M[i,k] = \ocompr{m}{((I,J),m)\from M,\,I=i,\,J=k}\\
N[k,j] = \ocompr{n}{((I,J),n)\from N,\,I=k,\,J=j},
\end{array}\]
the product $M[i,k]*N[k,j]$ is equal to:
\[\begin{array}{l}
\ocompr{x*y}{x\from\ocompr{m}{((I,J),m)\from M,\,I=i,\,J=k},\\
\skiptext{\comprr{x*y}}\!y\from\ocompr{n}{((I,J),n)\from N,\,I=k,\,J=j}},
\end{array}\]
which is normalized as follows\short{ by unnesting nested comprehensions:}\extended{ using Rule~(\ref{nested-comps}), after
some variable renaming}:
\[\begin{array}{l}
\ocompr{x*y}{((I,J),m)\from M,\,I=i,\,J=k,\,\textbf{let}\;x=m,\\
\skiptext{\comprr{x*y}}\!((I',J'),n)\from N,\,I'=k,\,J'=j,\,\textbf{let}\;y=n}\\
=\ocompr{m*n}{((I,J),m)\from M,\,I=i,\,J=k,\\
\skiptext{=\comprr{m*n}}\!((I',J'),n)\from N,\,I'=k,\,J'=j}.
\end{array}\]
Lastly, since the value of $e$ in the assignment $R[i,j] := e$ is lifted to a bag,
this assignment is translated to $R:=R\lhd\compr{((i,j),v)}{v\from e}$,
that is, $R$ is augmented with an indexed set that results from accessing the lifted
value of $e$. If $e$ contains a value, the comprehension
will return a singleton bag, which will replace $R[i,j]$ with that value.
After substituting the value $e$ with the term derived for $M[i,k]*N[k,j]$,
we get an assignment equivalent to the assignment~(\ref{assign:example}).

\subsection{Handling Array Updates in a Loop}\label{array-updates}

We now address the problem of translating array updates in a loop.
We classify updates into two categories:
\begin{enumerate}
\item Incremental updates of the form $d:=d\oplus e$, for some commutative
operation $\oplus$, where $d$ is an update destination, which is also
repeated as the left operand of $\oplus$.  It can also be written as
$d\opluseq e$. For example, $V[i]\pluseq{+}1$ increments $V[i]$ by 1.
\item All other updates of the form $d:=e$.
\end{enumerate}
Consider the following loop with a non-incremental update:
\begin{align}
&\textbf{for}\;i=1,N\;\textbf{do}\;V[g(i)]:=W[f(i)]\label{for-loop}
\end{align}
for some vectors $V$ and $W$, and some terms $f(i)$ and $g(i)$ that depend on
the index $i$.  Our framework translates this loop to an update to the
vector $V$, where all the elements of $V$ are updated at once, in a
parallel fashion:
\begin{align}
&V:=V\lhd\compr{(g(i),v)}{i\from\mathrm{range}(1,N),\label{update-compr1}\\
&\skiptext{V:=V\lhd\comprr{(g(i),v)}}(k,v)\from V,\,k=f(i)}.\nonumber
\end{align}
But this expression may not produce the same vector $V$ as the original
loop if there are recurrences in the loop, such as, when
the loop body is $V[i]:=V[i-1]$.
Furthermore, the join between range$(1,N)$ and $W$ in~(\ref{update-compr1})
looks unnecessary. We will transform such joins to array traversals in Section~\ref{remove-loops}.

In our framework, for-loops are embedded as generators inside the
comprehensions that are associated with the loop
assignments. Consider, for example, matrix copying:
\begin{align*}
&\mathbf{for}\;i=1,10\;\mathbf{do}\;\mathbf{for}\;j=1,20\;\mathbf{do}\;M[i,j]:=N[i,j].
\end{align*}
Using the translation of the assignment $M[i,j]:=N[i,j]$, the loop becomes:
\begin{align}
&\mathbf{for}\;i=1,10\;\mathbf{do}\label{loop-trans}\\
&\hspace*{4ex}\mathbf{for}\;j=1,20\;\mathbf{do}\nonumber\\
&\hspace*{8ex}M:=M\lhd\compr{((i,j),n)}{((I,J),n)\from N,\nonumber\\
&\hspace*{8ex}\skiptext{M:=M\lhd\comprr{((i,j),n)}}I=i,\,J=j}.\nonumber
\end{align}
To parallelize this loop, we embed the for-loops inside the
comprehension as generators:
\begin{align}
&M:=M\lhd\compr{((i,j),n)}{i\from\mathrm{range}(1,10),\label{embedded}\\
&\skiptext{M:=M\lhd\comprr{((i,j),n)}}j\from\mathrm{range}(1,20),\nonumber\\
&\skiptext{MM\comprr{((i,j),n)}}((I,J),n)\from N,\,I=i,\,J=j}.\nonumber
\end{align}
Notice the difference between the loop~(\ref{loop-trans}) and the assignment~(\ref{embedded}).
The former will do 10*20 updates to $M$ while the latter will only
do one bulk update that will replace all $M[i,j]$ with $N[i,j]$ at once.
This transformation can only apply when there are no recurrences across iterations.
\ignore{Also, if the body of a for-loop is a block of statements, we want
to embed the generator of this for-loop inside each comprehension derived
from every assignment in the block of statements.
This means that we want $\mathbf{for}\;i=\ldots\;\mathbf{do}\;\{\,s_1;\,s_2\,\}$
to be equal to
$\{\,\mathbf{for}\;i=\ldots\;\mathbf{do}\;s_1;\;\mathbf{for}\;i=\ldots\;\mathbf{do}\;s_2\,\}$.
}

\subsection{Eliminating Loop Iterations}\label{remove-loops}

Before we present the details of program translation, we address the problem of
eliminating index iterations, such as
$\mathrm{range(1,N)}$ in assignment~(\ref{update-compr1}),
and $\mathrm{range(1,10)}$ and range(1, 20) in assignment~(\ref{embedded}).
If there is a right inverse $F$ of $f$ such that $f(F(k))=k$,
then the assignment~(\ref{update-compr1}) is optimized to:
\begin{align}
&V:=V\lhd\compr{(g(F(k)),v)}{(k,v)\from W,\label{update-compr2}\\
&\skiptext{V:=V\lhd\comprr{(g(F(k)),v)}}\mathrm{inRange}(F(k),1,N)},\nonumber
\end{align}
where the predicate $\mathrm{inRange}(F(k),1,N)$ returns true if
$F(k)$ is within the range $[1,N]$. Given that the right-hand side of an
update may involve multiple array accesses, we can choose one whose index
term can be inverted. For example, for $V[i-1]$, the inverse of $k=i-1$ is $i=k+1$.
In the case where no such inverse can be derived,
the range iteration simply remains as is. One such example is
the loop $\mathbf{for}\;i=1,N\;\mathbf{do}\;V[i]:=0$,
which is translated to $V:=V\lhd\compr{(i,0)}{i\from\mathrm{range}(1,N)}$.

\subsection{Handling Incremental Updates}\label{handling-recurrences}

There is an important class of recurrences in loops that can be
parallelized using group-by and aggregation. Consider, for example, the following
loop with an incremental update:
\begin{align}
&\textbf{for}\;i=1,N\;\textbf{do}\;V[g(i)]\pluseq{+}W[i].\label{incr-loop}
\end{align}
\short{It}\extended{Let's say, for example, that there are 3 indexes overall, $i_1$, $i_2$, and $i_3$,
that have the same image under $g$, ie, $k=g(i_1)=g(i_2)=g(i_3)$.
Then, $V[k]$ must be set to $V[k]+W[i_1]+W[i_2]+W[i_3]$. In general,
we need to bring together all values of $W$ whose indexes have the same image under $g$.
That is, we need to group by $g(i)$.
Hence, the loop} can be translated to a comprehension with a group-by:
\begin{align*}
&V:=V\lhd\compr{(k,v+(+/w))}{(i,w)\from W,\,\mathrm{inRange}(i,1,N),\\
&\hspace*{15ex}\textbf{group\ by}\;k: g(i),\,(j,v)\from V,\,j=k},
\end{align*}
which groups $W$ by the destination index $g(i)$ and, for each group, it calculates the aggregation
$+/w$ of all values $w=W[i]$ with the same $g(i)$ value, but also adds the original value $v=V[g(i)]$ before the group-by.

If the destination of the incremental update is a variable, such as in $n\pluseq{+}W[i]$,
then the group-by is over $(\,)$, since there are no indexes used in $n$:
\begin{align*}
&n:=\ocompr{n+(+/w)}{(i,w)\from W,\,\mathrm{inRange}(i,1,N),\\
&\skiptext{n:=\comprr{n+(+/w)}}\textbf{group\ by}\;k: (\,)}.
\end{align*}
This group-by can be eliminated because
it forms a single group; in which case the variable $w$ is lifted to a bag that contains all the values of $W$:
\begin{align*}
&n:=\some{n+(+/\compr{w}{(i,w)\from W,\,\mathrm{inRange}(i,1,N)})}.
\end{align*}
We will discuss optimizations like this in Section~\ref{optimizations}.

\begin{figure*}
\framebox{
\begin{minipage}[l]{6.78in}
\begin{subequations}
\begin{align}
\intertext{\underline{$\mathcal{E}\sem{e}:$\qquad Translate the expression $e$ to a comprehension term}}
\mathcal{E}\sem{V} & =\;\some{V}\label{t1}\\
\mathcal{E}\sem{e.A} & =\;\ocompr{v.A}{v\from\mathcal{E}\sem{e}}\label{t3}\\
\mathcal{E}\sem{V[e_1,\ldots,e_n]} & =\;\ocompr{v}{k_1\from\mathcal{E}\sem{e_1},\ldots,k_n\from\mathcal{E}\sem{e_n},\,((i_1,\ldots,i_n),v)\from V,\,i_1=k_1,\ldots,i_n=k_n}\label{t4}\\
\mathcal{E}\sem{e_1\star e_2} & =\;\ocompr{v_1\star v_2}{v_1\from\mathcal{E}\sem{e_1},\,v_2\from\mathcal{E}\sem{e_2}}\label{t5}\\
\mathcal{E}\sem{(e_1,\ldots,e_n)} & =\;\ocompr{(v_1,\ldots,v_n)}{v_1\from\mathcal{E}\sem{e_1},\,\ldots,\,v_n\from\mathcal{E}\sem{e_n}}\label{t6}\\
\mathcal{E}\sem{\lt\,A_1\req e_1,\ldots,A_n\req e_n\,\gt}
& =\;\ocompr{\lt\,A_1\req v_1,\ldots,A_n\req v_n\,\gt}{v_1\from\mathcal{E}\sem{e_1},\,\ldots,\,v_n\from\mathcal{E}\sem{e_n}}\label{t7}\\
\mathcal{E}\sem{const} & =\;\some{const}\label{t8}
\end{align}
\end{subequations}\\[-0.5ex]
\begin{minipage}[l]{2.65in}
\begin{subequations}
\begin{align}
\nonumber\\[-3.2em]
\intertext{\underline{$\mathcal{K}\sem{d}:$\qquad Derive the destination index from $d$}}
\mathcal{K}\sem{V} & =\;\some{(\,)}\label{k1}\\
\mathcal{K}\sem{d.A_i} & =\;\mathcal{K}\sem{d}\label{k3}\\
\mathcal{K}\sem{V[e_1,\ldots,e_n]} & =\;\mathcal{E}\sem{(e_1,\ldots,e_n)}\label{k4}
\end{align}
\end{subequations}
\end{minipage}\hspace*{5ex}
\begin{minipage}[l]{3.8in}
\begin{subequations}
\begin{align}
\nonumber\\[-3.2em]
\intertext{\underline{$\mathcal{D}\sem{d}(k):$\qquad Derive $d$ from the destination index $k$}}
\mathcal{D}\sem{V}(k) & =\;\some{V}\label{d1}\\
\mathcal{D}\sem{d.A_i}(k) & =\;\ocompr{v.A_i}{v\from\mathcal{D}\sem{d}(k)\label{d3}}\\
\mathcal{D}\sem{V[e_1,\ldots,e_n]}(k) & =\;\ocompr{v}{((i_1,\ldots,i_n),v)\from V,\hspace*{10ex}\nonumber\\
&\skiptext{=\;\,\comprr{v}}(i_1,\ldots,i_n)=k}\label{d4}
\end{align}
\end{subequations}
\end{minipage}\\[1.6ex]
\begin{subequations}
\begin{align}
\nonumber\\[-3.2em]
\intertext{\underline{$\mathcal{U}\sem{d}(x):$\qquad Update the destination $d$ with the value $x$}}
\mathcal{U}\sem{V}(x) & =\;\liste{V:=\ocompr{v}{(k,v)\from x}}\label{u9}\\
\mathcal{U}\sem{d.A_i}(x) & =\;\mathcal{U}\sem{d}(\compr{(k,\lt\,A_1\req w.A_1,\ldots,A_i\req v,\ldots,A_n\req w.A_n\,\gt)}{(k,v)\from x,\,w\from\mathcal{D}\sem{d}(k)})\label{u11}\\
\mathcal{U}\sem{V[e_1,\ldots,e_n]}(x) & =\;\liste{V:=V\lhd x}\label{u12}
\end{align}
\end{subequations}
\begin{subequations}
\begin{align}
\nonumber\\[-2.5em]
\intertext{\underline{$\mathcal{S}\sem{s}(\bar{q}):$\qquad Translate the statement $s$ to a target code block using the list of for-loop qualifiers $\bar{q}$}}
\mathcal{S}\sem{d\,\opluseq e}(\bar{q})
& =\;\mathcal{U}\sem{d}(\compr{(k,w\oplus(\oplus/v))}{\bar{q},\,v\from\mathcal{E}\sem{e},\,k\from\mathcal{K}\sem{d},\nonumber\\
&\skiptext{ =\;\mathcal{U}\sem{d}(\comprr{(k,w\oplus(\oplus/v))}}\;\mathbf{group\ by}\,k,\,w\from\mathcal{D}\sem{d}(k)})\label{s1}\\
\mathcal{S}\sem{d:=e}(\bar{q}) & =\;\mathcal{U}\sem{d}(\compr{(k,v)}{\bar{q},\,v\from\mathcal{E}\sem{e},\,k\from\mathcal{K}\sem{d}})\label{s2}\\
\mathcal{S}\sem{\mathbf{var}\;V:t=e}(\bar{q})
& =\;\mathcal{S}\sem{V:=e}(\bar{q})\label{x18}\\
\mathcal{S}\sem{\mathbf{for}\;v=e_1,e_2\;\mathbf{do}\;s}(\bar{q})
& =\;\mathcal{S}\sem{s}(\bar{q}\app\liste{v_1\from\mathcal{E}\sem{e_1},\,v_2\from\mathcal{E}\sem{e_2},\,v\from\mathrm{range}(v_1,v_2)})\label{x19}\\
\mathcal{S}\sem{\mathbf{for}\;v\;\mathbf{in}\;e\;\mathbf{do}\;s}(\bar{q})
& =\;\mathcal{S}\sem{s}(\bar{q}\app\liste{A\from\mathcal{E}\sem{e},\,(i,v)\from A})\label{x20}\\
\mathcal{S}\sem{\mathbf{while}\;(e)\;s}(\bar{q})
& =\;\liste{\mathrm{while}(\mathcal{E}\sem{e},\mathcal{S}\sem{s}(\bar{q}))}\label{x21}\\
\mathcal{S}\sem{\mathbf{if}\;(e)\;s_1\;\mathbf{else}\;s_2}(\bar{q})
& =\;\mathcal{S}\sem{s_1}(\bar{q}\app\liste{p\from\mathcal{E}\sem{e},\,p})
\app\mathcal{S}\sem{s_2}(\bar{q}\app\liste{p\from\mathcal{E}\sem{e},\,!p})\label{x22}\\
\mathcal{S}\sem{\{\,s_1;\,\ldots;\,s_n\}}(\bar{q})
& =\;\mathcal{S}\sem{s_1}(\bar{q})\app\cdots\app\mathcal{S}\sem{s_n}(\bar{q})\label{x23}
\end{align}
\end{subequations}
\end{minipage}
}
\caption{Rules for translating loop-based programs to target code}\label{translation-rules}
\end{figure*}

\subsection{Program Translation}\label{translation}

The target of our translations is a list of statements, where a statement $c$ has the following syntax:
\[\begin{array}{rcll}
\multicolumn{4}{l}{\mbox{\textbf{Target Code}:}}\\
c & ::= & v:=e & \mbox{assignment}\\
&|& \mathrm{while}(e,c) & \mbox{loop}\\
&|& \liste{c_1,\ldots,c_n} & \mbox{code block}
\end{array}\]
In the target code, an assignment to a variable $v$ of type $t$ gets a
value $e$ of type $\bag{t}$. An assignment to an array is done in
bulk, by replacing the entire array with a new one. The while-loop
corresponds to the while statement in Figure~\ref{syntax}; it repeats
the code $c$ in its body while the condition $e$ is true. Finally, a
code block is like a block of statements that need to be evaluated in
order.

The rules for translating loop-based programs to the target code are
given in Figure~\ref{translation-rules}. They are mainly given in terms of the
semantic functions $\mathcal{E}$ and $\mathcal{S}$ that translate
expressions and statements, respectively. The syntactic brackets
$\sem{\ldots}$ enclose syntactic elements, as defined in
Figure~\ref{syntax}. The rules for $\mathcal{E}\sem{e}$, given in
Equations~(\ref{t1})-(\ref{t8}), translate an expression $e$ of type $t$ to
a comprehension term of type $\bag{t}$. For example, using
Equation~(\ref{t4}), $M[1,2]$ is translated to:
\begin{align*}
&\ocompr{v}{k\from\some{1},\,l\from\some{2},\,((i,j),v)\from M,\,i=k,\,j=l}\\
&=\ocompr{v}{((i,j),v)\from M,\,i=1,\,j=2}.
\end{align*}
The rules for $\mathcal{S}\sem{s}(\bar{q})$, given in
Equations~(\ref{s1})-(\ref{x23}), translate a statement $s$ to a list
of target code statements. $\mathcal{S}\sem{s}(\bar{q})$ is
parameterized by a list of qualifiers $\bar{q}$ that correspond to the
for-loop iterations, to be embedded in the comprehensions derived from
the assignments in the loop body. This is always possible because of
Theorem~\ref{for-distribution}. That is, the for-loops in
Equations~(\ref{x19}) and~(\ref{x20}) become qualifiers, which are
propagated to the translation of their body $s$ along with the current
$\bar{q}$ (where $\app$ is list concatenation). While-loops, on the
other hand, are translated to while-loop target statements in
Equation~(\ref{x21}) because they are not parallelized. The
qualifiers $\bar{q}$ are propagated to every statement in a block, as
shown in Equation~(\ref{x23}). Equations~(\ref{s1}) and~(\ref{s2})
translate assignments. An incremental update $d\,\opluseq e$, equal
to $d:=d\oplus e$, is translated by Equation~(\ref{s1}). All other
assignments are translated by Equation~(\ref{s2}). Both
Equations~(\ref{s1}) and~(\ref{s2}) use the semantic function
$\mathcal{K}$ that derives the destination indexes of the
assignment, and the semantic function $\mathcal{U}$ that generates the
update associated with the assignment. More specifically,
$\mathcal{U}\sem{d}(x)$ replaces the destination $d$ with the value
$x$ by reconstructing the destination variable from its components,
replacing the components reachable from $d$ with $x$. For example,
$\mathcal{U}\sem{V[1]}(\bag{(1,10)})$, which is equal to
$\liste{V:=V\lhd\bag{(1,10)}}$, updates $V$ to be equal to $V$ but
with $V[1]$ replaced with $10$. The incremental update $d\,\opluseq e$
is translated by Equation~(\ref{s1}) to a comprehension with a
group-by over the destination index $d$ and an aggregation $\oplus/v$
of all $e$ values associated with the same group-by key. The value $w$
added to the aggregation is the initial value of $d$ before the
loop. This value cannot be computed from $\mathcal{E}\sem{d}$ because
it is correlated to the destination index $k$. Instead, it is derived
from $k$ using the semantic function $\mathcal{D}$.

\extended{Theorem~\ref{correctness-proof} in Appendix~\ref{proofs} proves that
the transformation rules in Figure~\ref{translation-rules} under the restrictions in Definition~\ref{PFOR-def} are meaning preserving.}
\ignore{In the extended version of this paper~\cite{diablo-extended}, we provide a proof that
the transformation rules in Figure~\ref{translation-rules}
under the restrictions in Definition~\ref{PFOR-def} are meaning preserving.}

\subsection{Examples of Program Translation}

\extended{First consider the following statement $s$ that consists of a non-incremental update in a for-loop:}
\short{Consider the following statement $s$:}
\[\begin{array}{l}
\mathbf{for}\;i=1,10\;\mathbf{do}\;V[i]:=W[i].
\end{array}\]
It is translated as follows:
\[\begin{array}{l}
\mathcal{S}\sem{s}(\liste{\,})\\
\why{from Equation~(\ref{x19})}
=\mathcal{S}\sem{V[i]:=W[i]}(\liste{v_1\from\mathcal{E}\sem{1},\,v_2\from\mathcal{E}\sem{10},\\
\skiptext{=\mathcal{S}\sem{V[i]:=W[i]}([\,}i\from\mathrm{range}(v_1,v_2)})\\
\why{using Equation~(\ref{t8}) and after normalization}
=\mathcal{S}\sem{V[i]:=W[i]}(\liste{i\from\mathrm{range}(1,10)})\\
\why{from Equation~(\ref{s2})}
=\mathcal{U}\sem{V[i]}(\compr{(k,v)}{i\from\mathrm{range}(1,10),\,v\from\mathcal{E}\sem{W[i]},\\
\skiptext{=\mathcal{U}\sem{V[i]}(\comprr{(k,v)}}k\from\mathcal{K}\sem{V[i]}})\\
\why{from Equations~(\ref{t4}) and~(\ref{k4})}
=\mathcal{U}\sem{V[i]}(\compr{(k,v)}{i\from\mathrm{range}(1,10),\\
\skiptext{\mathcal{U}\sem{V[i]}(\comprr{()}}v\from\ocompr{w}{(j,w)\from W,\,j=i},\,k\from\some{i}})\\
\why{after normalization}
=\mathcal{U}\sem{V[i]}(\compr{(i,w)}{i\from\mathrm{range}(1,10),\,(j,w)\from W,\,j=i})\\
\why{from Equation~(\ref{u12})}
= \liste{V:=V\lhd\compr{(i,w)}{i\from\mathrm{range}(1,10),\,(j,w)\from W,\,j=i}}\\
\why{after eliminating the loop iteration}
= \liste{V:=V\lhd\compr{(i,w)}{(i,w)\from W,\,\mathrm{inRange}(i,1,10)}}.
\end{array}\]
Note that, the assignment $V:=V\lhd\ldots$ is done in parallel,
such as replacing an RDD with another RDD in Spark.
Consider the following loop $s$ with an incremental update:
\[\begin{array}{l}
\mathbf{for}\;i=1,10\;\mathbf{do}\;W[K[i]]\pluseq{+}V[i].
\end{array}\]
Then, from Equation~(\ref{x19}), $\mathcal{S}\sem{s}(\liste{\,})$ is equal to:
\[\begin{array}{l}
\mathcal{S}\sem{W[K[i]]\pluseq{+}V[i]}(\liste{v_1\from\mathcal{E}\sem{1},\,v_2\from\mathcal{E}\sem{10},\\
\skiptext{\mathcal{S}\sem{W[K[i]]\pluseq{+}V[i]}([\,}i\from\mathrm{range}(v_1,v_2)})\\
=\mathcal{S}\sem{W[K[i]]\pluseq{+}V[i]}(\liste{i\from\mathrm{range}(1,10)}).
\end{array}\]
To translate $W[K[i]]\pluseq{+}V[i]$ using Equation~(\ref{s1}), we need to derive
the destination index using Equation~(\ref{k4}):
\[\begin{array}{l}
\mathcal{K}\sem{W[K[i]]} = \mathcal{E}\sem{K[i]}
=\ocompr{a}{(m,a)\from K,\,m=i}
\end{array}\]
and the destination value from the destination index using Equation~(\ref{d4}):
\[\begin{array}{l}
\mathcal{D}\sem{W[K[i]]}(k)
=\ocompr{v}{(i,v)\from W,\,i=k}.
\end{array}\]
Hence, the loop translation is:
\[\begin{array}{l}
\mathcal{S}\sem{W[K[i]]\pluseq{+}V[i]}(\liste{i\from\mathrm{range}(1,10)})\\
\why{from Equation~(\ref{s1})}
=\mathcal{U}\sem{W[K[i]]}(\compr{(k,w+(+/v))}{i\from\mathrm{range}(1,10),\\
\hspace*{5ex}(l,v)\from V,\,l=i,\,k\from\mathcal{K}\sem{W[K[i]]},\\
\hspace*{5ex}\mathbf{group\ by}\;k,\,w\from\mathcal{D}\sem{W[K[i]]}(k)})
\end{array}\]
\[\begin{array}{l}
=\mathcal{U}\sem{W[K[i]]}(\compr{(k,w+(+/v))}{i\from\mathrm{range}(1,10),\\
\hspace*{5ex}(l,v)\from V,\,l=i,\,k\from\ocompr{a}{(m,a)\from K,\,m=i},\\
\hspace*{5ex}\mathbf{group\ by}\;k,\,w\from\ocompr{v}{(i,v)\from W,\,i=k}})\\
=\mathcal{U}\sem{W[K[i]]}(\compr{(k,w+(+/v))}{i\from\mathrm{range}(1,10),\\
\hspace*{5ex}(l,v)\from V,\,l=i,\,(m,a)\from K,\,m=i,\\
\hspace*{5ex}\mathbf{group\ by}\;a,\,(j,w)\from W,\,j=a})\\
\why{from Equation~(\ref{u12})}
=\liste{W:=W\lhd\compr{w+(+/v)}{i\from\mathrm{range}(1,10),\\
\hspace*{10ex}(l,v)\from V,\,l=i,\,(m,a)\from K,\,m=i,\\
\hspace*{10ex}\mathbf{group\ by}\;a,\,(j,w)\from W,\,j=a}},
\end{array}\]
which is optimized to the following target code after removing the loop iteration:
\[\begin{array}{l}
\liste{W:=W\lhd\compr{w+(+/v)}{(i,v)\from V,\\
\hspace*{10ex}\mathrm{inRange}(i,1,10),\,(m,a)\from K,\,m=i,\\
\hspace*{10ex}\mathbf{group\ by}\;a,\,(j,w)\from W,\,j=a}}.
\end{array}\]
\ \\

\section{Optimizations}\label{optimizations}

As discussed in Section~\ref{framework}, incremental updates
on variables of a basic type, such as  $n\pluseq{+}W[i]$, can be translated to total aggregations.
This translation is actually an optimization of the default translation.
The optimization rule, for a constant group-by key $c$, is:
\begin{align}
&\compr{e}{\bar{q_1},\,\mathbf{group\ by}\;p:c,\,\bar{q_2}}\label{constant-key}\\
&\rightarrow\;\compr{e}{\mathbf{let}\;p=c,\,\forall v_i:\mathbf{let}\;v_i=\compr{v_i}{\bar{q_1}},\,\bar{q_2}}\nonumber
\end{align}
where $v_i$ are the pattern variables in $\bar{q_1}$.
For example, consider the assignment $n\pluseq{+}W[i]$, which is translated to:
\begin{align*}
&n:=\ocompr{n+(+/w)}{(i,w)\from W,\,\textbf{group\ by}\;k: (\,)}.
\end{align*}
The right-hand side of this assignment is optimized to:
\begin{align*}
&\ocompr{n+(+/w)}{\textbf{let}\;k=(\,),\,\textbf{let}\;w=\compr{w}{(i,w)\from W}}\\
&=\;\some{n+(+/\compr{w}{(i,w)\from W})}
\end{align*}
which is more efficient because it does not use a group-by.
The same happens when indexes in the destination are constants, such as in $M[1,2]\pluseq{+}1$.
Then, the group-by on $(1,2)$ can be removed using Rule~(\ref{constant-key}):
\begin{align*}
&M\lhd\compr{(k,v+(+/c)}{\textbf{let}\;c=1,\,\textbf{group\ by}\;k: (1,2),\\
&\skiptext{M\lhd\comprr{(k,v+(+/c)}}((i,j),v)\from M,\,i=1,\,j=2}\\
&=\;M\lhd\compr{(k,v+(+/c)}{\textbf{let}\;k=(1,2),\\
&\skiptext{=\;\,M\lhd\comprr{(k,v+(+/c)}}\textbf{let}\;c=\compr{c}{\textbf{let}\;c=1},\\
&\skiptext{=\;\,M\lhd\comprr{(k,v+(+/c)}}((i,j),v)\from M,\,i=1,\,j=2}\\
&=\;M\lhd\compr{((1,2),v+1)}{((i,j),v)\from M,\,i=1,\,j=2}
\end{align*}

Another optimization is when the group-by key is unique, that is, when the group-by function
is injective. In that case, each group is a singleton bag. The group-by can be eliminated using the following rule:
\begin{align}
&\compr{e}{\bar{q_1},\,\mathbf{group\ by}\;p:k,\,\bar{q_2}}\label{unique-key}\\
&\rightarrow\;\compr{e}{\bar{q_1},\;\mathbf{let}\;p=k,\,\forall v_i:\mathbf{let}\;v_i=\bag{v_i},\,\bar{q_2}}\nonumber
\end{align}
That is, the group-by is removed and every pattern variable $v_i$ in $\bar{q_1}$
is lifted to a singleton bag that represents the group, that is, it contains $v_i$ only.
For example, the loop:
\begin{align*}
&\textbf{for}\;i=1,10\;\textbf{do}\;V[i]\pluseq{+} W[i]
\end{align*}
has a default translation, after removing the for-loop:
\begin{align*}
&V\lhd\compr{(k,v+(+/w))}{(i,w)\from W,\,\textrm{inRange}(i,1,10),\\
&\skiptext{\comprr{(k,v+(+/w))}}\textbf{group\ by}\;k:i,\,(j,v)\from V,\;j=k}
\end{align*}
Here, the group-by key is unique since it is the index of $W$. Based on Rule~(\ref{unique-key}),
this term is optimized to:
\begin{align*}
&V\lhd\compr{(k,v+(+/w))}{(i,w)\from W,\,\textrm{inRange}(i,1,10),\\
&\skiptext{M\lhd\comprr{k}}\textbf{let}\;k=i,\,\textbf{let}\;w=\bag{w},\,(j,v)\from V,\;j=k}\\
&=V\lhd\compr{(i,v+w)}{(i,w)\from W,\,\textrm{inRange}(i,1,10),\\
&\skiptext{=M\lhd\comprr{(i,v+w)}}\,(j,v)\from V,\;j=i}
\end{align*}
Inferring whether a group-by key is unique is similar to inferring
whether an assignment destination is affine
(Section~\ref{recurrences}). A generator $(i,w)\from W$ for
an array $W$ indicates that $i$ is unique. If the group-by key
is an affine term that consists of all array indexes in the generators before the group-by,
then it is a unique key.

\ignore{
Finally, updates to an array $V$ are translated to the target code $\liste{V:=V\lhd x}$.
In some cases, this code can be optimized to $\liste{V:=x}$, which does not require the
join $\lhd$. This is only possible if the existing $V$ has indexes that are a subset or equal
to the indexes in $x$.
}

\extended{
\section{Packing/Unpacking Arrays}\label{pack-unpack}

In our framework, sparse arrays are an abstract representation of real
arrays that may have been stored and partitioned into various custom
dense storage structures. This separation of representation from
implementation simplifies the language semantics by abstracting the
implementation details from programs. More importantly, it makes
easier to change the implementation without changing the programs. But
this separation may introduce one more level of interpretation needed
for restructuring data when loading storage structures into sparse
arrays (unpacking) and storing sparse arrays to storage structures
(packing). Our framework though can remove this extra layer of
interpretation without any fundamental extension to the framework. In
this paper, we discuss only matrices.

In our framework, a sparse matrix of type matrix$[T]$ is represented
as $\bag{((\mathrm{long},\mathrm{long}),T)}$, which contains the
matrix elements in sparse form. One example of a concrete
implementation of a matrix is organizing the matrix elements into
equal sized chunks, called
tiles~\cite{tiledb,scidb:sigmod10,arraystore:sigmod11}, where each
tile is a dense array of elements. This is called a tiled matrix. One
possible implementation of a tiled matrix is
$\bag{((\mathrm{long},\mathrm{long}),\mathrm{Array}[T])}$, which is a
bag of tiles where each tile has an upper-left coordinate index and a
$\mathrm{Array}[T]$ which is a dense vector that contains the matrix
elements that belong to this tile. The translation from dense to sparse
vector, and vice versa, can be defined as follows in Scala:
\begin{tabbing}
def scan(V) = V.zipWithIndex.map(\_.swap)\\
def form(L,n) = \{ \=\+val a = new Array(n)\\
                  for( (i,v) $\leq$- if i$\geq$0 \&\& i$\leq$n ) a(i)=v\\
                  a \}
\end{tabbing}
where scan$(V)$ converts the dense vector $V=\liste{v_1,\ldots,v_n}$
into the sparse vector $\bag{(0,v_1),$ $\ldots,(n-1,v_n)}$ and
form$(L,n)$ converts the sparse vector $L$ to a dense vector of size
$n$.  A tile is the unit of distributed processing. If, in addition,
we use a Scala parallel collection, such as ParArray, to store the
dense vector in a tile, then we would have thread-level data
parallelism along with distributed data parallelism.

Suppose that the tiles are of size $n*m$.
We can map a tiled
matrix $N$ to a sparse matrix using the function unpack$(N)$:
\begin{align*}
&\compr{(\, (I+k/m,\, J+k\%m),\, v\, )}{((I,J),L)\from N,\\
&\skiptext{\comprr{(\, (I+k/m,\, J+k\%m),\, v\, )}}(k,v)\from\mathrm{scan}(L)}
\end{align*}
We can map a sparse matrix $M$
to a tiled matrix using the function pack$(M)$:
\begin{align*}
&\compr{(\, (I*n,\, J*m),\, \mathrm{form}(z,n*m)\, )}{((i,j),v)\from M,\\
&\hspace*{4ex}\mathbf{let}\;z=(i+j*n,v),\,\mathbf{group\  by}\; (I: i/n,J:j/m)}
\end{align*}
Under these mappings, $N[i,j]$ is translated to:
\begin{align*}
&\ocompr{v}{((I,J),v)\from\mathrm{unpack}(N),\,I=i,\,J=j}\\
&= \ocompr{v}{((I,J),L)\from N,\, (k,v)\from\mathrm{scan}(L),\,I=i,\,J=j}
\end{align*}
which directly traverses the data in the matrix tiles.  Assignments to
a matrix $N$, which are normally translated to $N:=N\lhd x$ by
Equation~(\ref{u12}), are now translated to
$N:=\mathrm{pack}(\mathrm{unpack}(N)\lhd x)$.  This expensive
unpacking and packing of $N$ can be removed by transforming this
assignment to $N:=N\lhd' \mathrm{pack}(x)$, where $\lhd'$ merges two
tiled matrices.  The tile merging in $N\lhd' \mathrm{pack}(x)$ can be
implemented without shuffling if we keep every matrix partitioned by
the tile coordinates and set the group-by partitioner in
$\mathrm{pack}(x)$ to be equal to the $N$ partitioner.  Then the tile
merging can be implemented using zipPartitions in Spark, which does
not require any shuffling.
}

\begin{figure*}
\begin{center}
\hspace*{-1ex}\scalebox{0.67}{\includegraphics{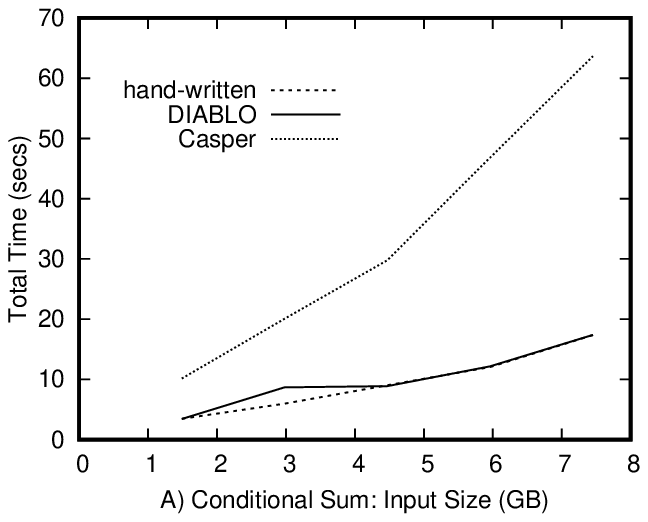}}
\hspace*{-3.8ex}\scalebox{0.67}{\includegraphics{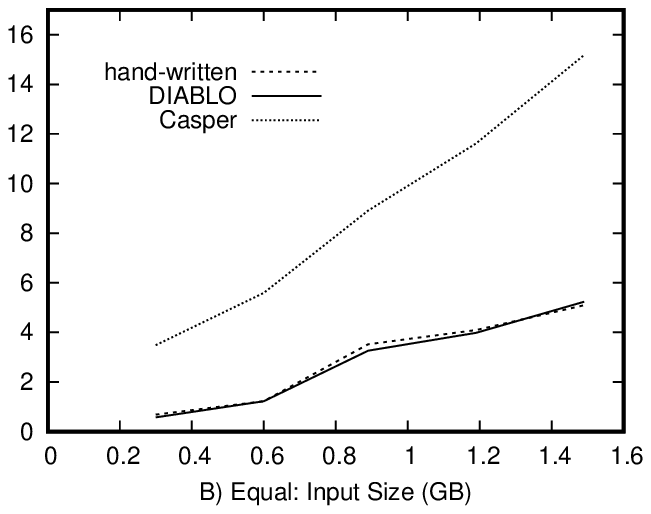}}
\hspace*{-3.8ex}\scalebox{0.67}{\includegraphics{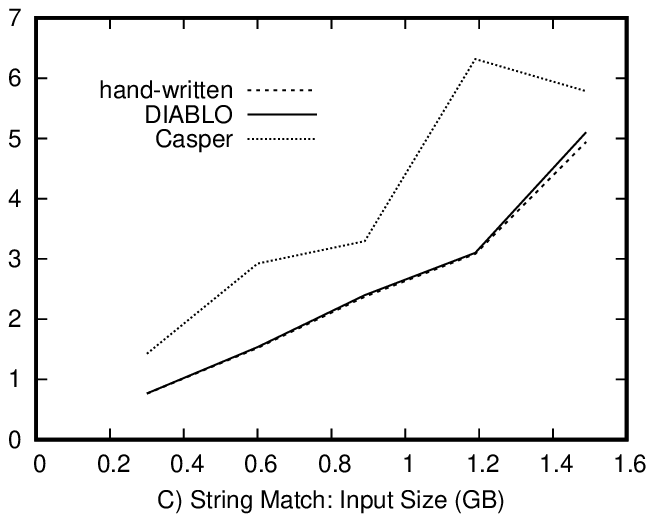}}
\hspace*{-3.8ex}\scalebox{0.67}{\includegraphics{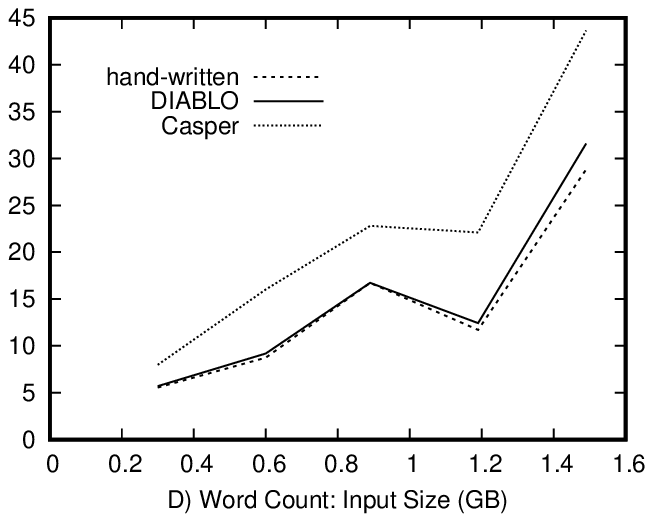}}\\
\hspace*{-1ex}\scalebox{0.67}{\includegraphics{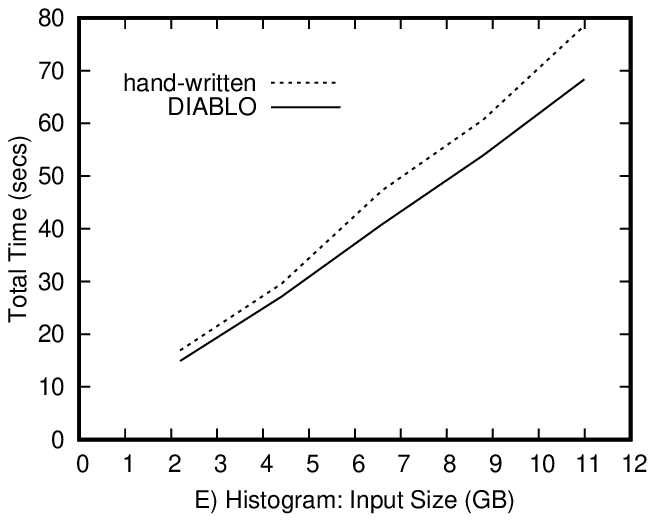}}
\hspace*{-3.8ex}\scalebox{0.67}{\includegraphics{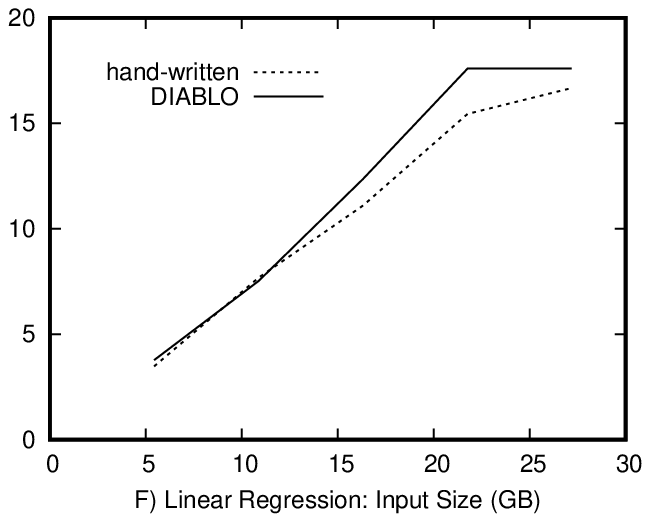}}
\hspace*{-3.8ex}\scalebox{0.67}{\includegraphics{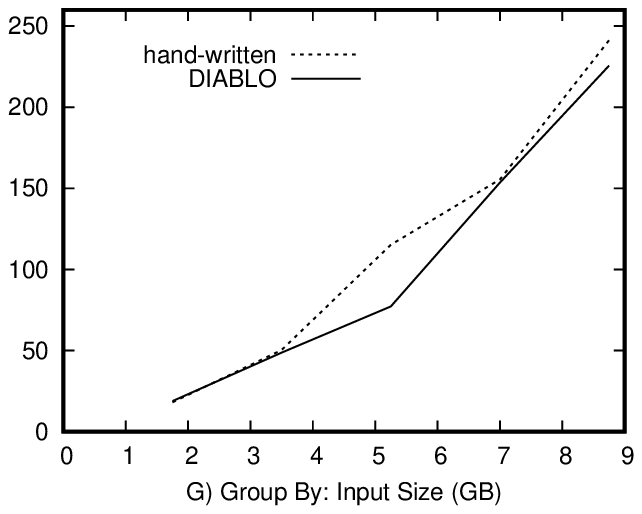}}
\hspace*{-3.8ex}\scalebox{0.67}{\includegraphics{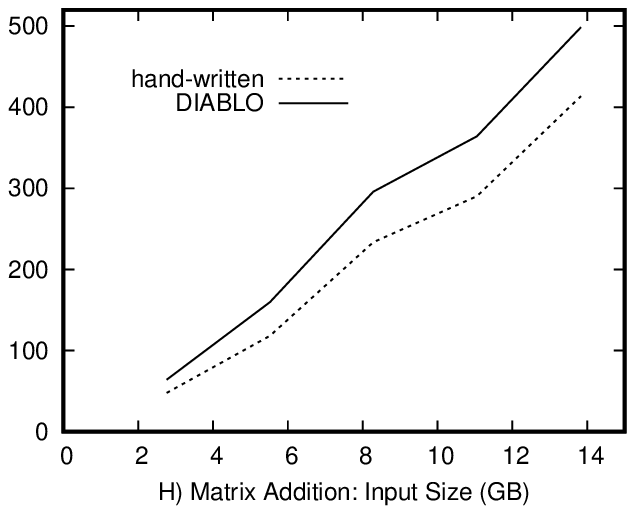}}\\
\hspace*{-1ex}\scalebox{0.67}{\includegraphics{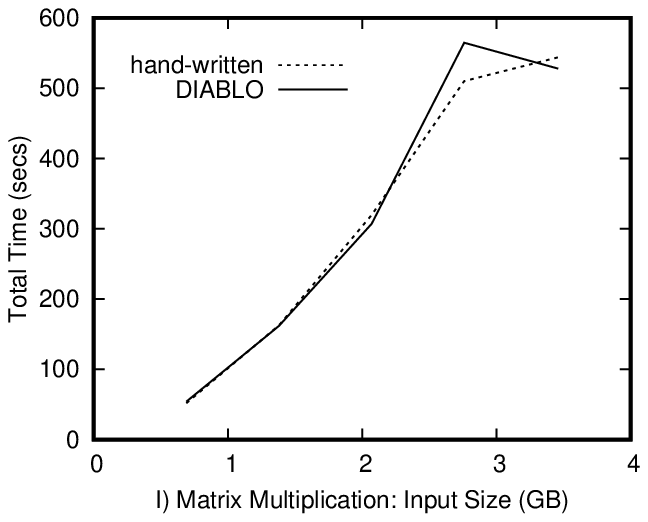}}
\hspace*{-3.8ex}\scalebox{0.67}{\includegraphics{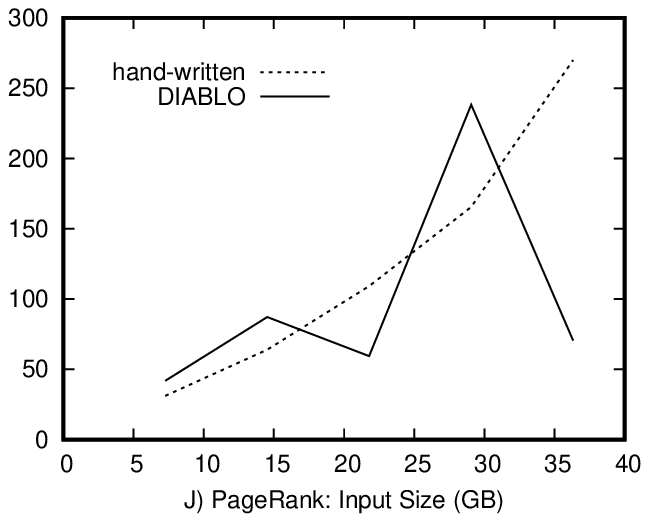}}
\hspace*{-3.8ex}\scalebox{0.67}{\includegraphics{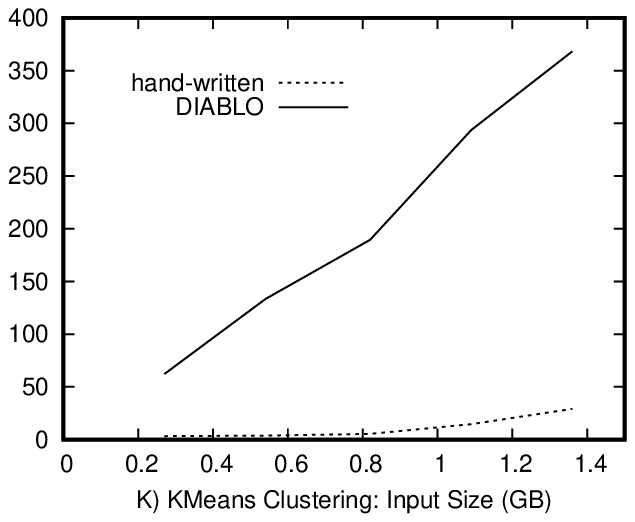}}
\hspace*{-3.8ex}\scalebox{0.67}{\includegraphics{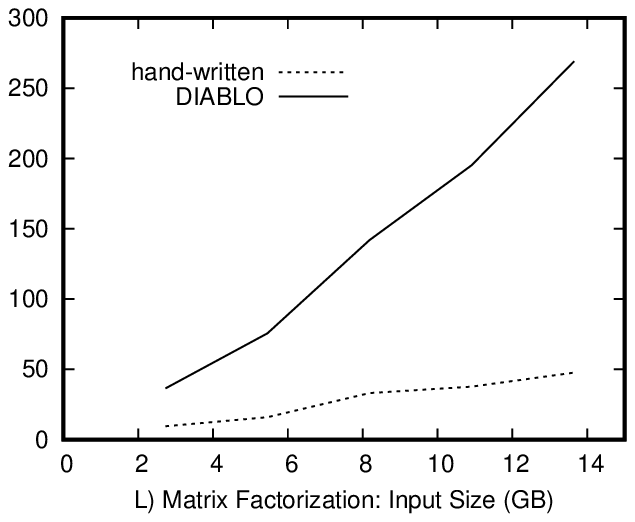}}
\end{center}
\caption{Performance of DIABLO relative to hand-written Spark code}\label{performance-results}
\end{figure*}

\section{Performance Evaluation}\label{performance}

DIABLO is implemented on top of DIQL~\cite{diql:BigData,diql}, which is a query
optimization framework that optimizes and compiles queries to Java
byte code at compile-time. DIQL can run on Apache Spark, Apache
Flink, Cascading/Scalding, and Scala Parallel collections. 
DIABLO compiles loop-based programs to monoid comprehensions, which in
turn are translated to byte code by the DIQL compiler.
DIABLO is currently implemented on Spark, Flink, and on Scala's
Parallel Collections.\extended{  The translator from loop-based programs to
monoid comprehensions is less than 500 lines of Scala code, while the
normalizer and optimizer of comprehensions are also less than 500
lines total.}
The DIABLO code is available as part of the DIQL source code on GitHub~\cite{diql}.
The subdirectory {\tt benchmarks/diablo}
in the source code contains all the benchmark programs, the scripts,
and the detailed execution logs with all the measurements derived from our
performance evaluations (the file README.md explains how to repeat these
experiments).

\extended{Our translation scheme is general, since it can translate any
loop-based program that satisfies our restrictions, and
efficient, since it uses simple program transformations, instead of
searching to match specific program templates.}  We first evaluated the
translator efficiency of DIABLO relative to MOLD~\cite{mold:oopsla14}
and \textsc{Casper}~\cite{casper:sigmod18} (Table~\ref{compilations}).
The programs used in these evaluations are described next in this section.
The translation times for MOLD were taken directly from the MOLD
paper~\cite{mold:oopsla14} but are not verified, because at the time
of writing, we could not install MOLD due to
software dependency issues.  In addition, although both the binaries and
source code of \textsc{Casper} are available at~\cite{casper:site}, we
were not able to validate some of the results reported
in~\cite{casper:sigmod18}.  More specifically, based on our
communication with the main developer of \textsc{Casper}, we tried
many configurations and libraries, but were not able to compile some of
the test files provided with the source code.  The results reported
here were run on Casper 0.1.1, with Sketch 1.7.5 and Dafny
1.9.7. These experiments were done on a 2.7 GHz Intel
Core i5 with 8GB RAM. Each program was run 4 times.
\textsc{Casper} was able to synthesize code for Histogram but its validator failed
to validate the code.  For Linear Regression, \textsc{Casper} was
taking too long so we had to abort it after 19 hours. The fail entries
in Table~\ref{compilations} are failures to
synthesize code for the test files; these errors were reported by the
Dafny program synthesizer.  We can see that the DIABLO translator is
far more efficient than both MOLD and \textsc{Casper} and, unlike
these systems, can translate complex programs. In
fact, \textsc{Casper} can only translate trivial flat loops.

\begin{table}
\begin{center}
{\scriptsize\begin{tabular}{|l||r|r|r|}\hline
test program & MOLD & \textsc{Casper} & DIABLO\\\hline
Average & & 172.25 & 5.75\\
Conditional Count & & 20.25 & 5.75\\
Conditional Sum & & 18.75 & 5.25\\
Count & & 9.75 & 5.75\\
Equal & & 11.25 & 5.75\\
Equal Frequency & & 778.00 & 5.75\\
String Match & 68 & 806.00 & 8.50\\
Sum & & 10.25 & 5.00\\
Word Count & 11 & 102.25 & 6.50\\
Histogram & 233 & 10272.00 & 9.00\\
Matrix Multiplication & 40 & fail & 8.25\\
Linear Regression & 28 & $>$19 hours & 8.75\\
KMeans & 340 & fail & 9.75\\
PCA & 66 & fail & 13.25\\
PageRank & & & 9.50\\
Matrix Factorization & & & 14.50\\\hline
\end{tabular}}
\end{center}
\caption{Compilation time in secs}\label{compilations}
\end{table}

Although the focus of our work is on distributed processing, not
shared-memory data parallelism, our second set of experiments was to
evaluate a variety of loop-based programs in two ways: in parallel
using Scala's parallel collections and sequentially using regular
lists.  That is, each one of these loop-based programs was compiled to
parallel and to sequential Scala programs, and these two programs were
evaluated over the same data.  Scala uses thread-level shared-memory
data parallelism on a multi-core computer to process parallel
collections. For these evaluations, we used one server with Xeon E5-2680v3
at 2.5GHz, with 24 cores and 128GB RAM.  The results,
shown in Table~\ref{parallel}, are based on the programs and data
described next in this section.  Each experiment was evaluated 4 times
and the mean value was used. We can see that for all, but
Group-By and KMeans, the DIABLO parallel programs are faster than
their sequential counterparts.

\begin{table}
\begin{center}
\scriptsize\begin{tabular}{|l||r|r|r|r|}\hline
test program & count & size (MB) & par & seq\\\hline
Conditional Sum & $10^9$ & 61035 & 19.6 & 40.6\\
Equal & $5\times 10^8$ & 20504 & 9.2 & 33.2\\
String Match & $5\times 10^8$ & 20504 & 8.3 & 32.6\\
Word Count & $5\times 10^7$ & 2050 & 57.1 & 69.4\\
Histogram & $5\times 10^7$ & 3338 & 8.2 & 30.6\\
Linear Regression & $10^8$ & 13924 & 13.5 & 39.0\\
Group-By & $5\times 10^7$ & 2766 & 56.6 & 51.9\\
Matrix Addition & $210\times 210$ & 10 & 0.13 & 216.0\\
Matrix Multiplication & $420\times 420$ & 39 & 20.8 & 137.8\\
PageRank & 1500000 & 279 & 10.9 & 44.9\\
KMeans & 500000 & 70 & 32.6 & 26.2\\
Matrix Factorization & $980\times 980$ & 210 & 13.2 & 22.7\\\hline
\end{tabular}
\end{center}
\caption{Parallel (par) vs Sequential (seq) evaluation time in secs}\label{parallel}
\end{table}

To evaluate the quality of our generated code on a distributed
platform, we have tested our system on 12 programs and compared their
evaluation against efficient hand-written programs on Spark.
The platform used for our evaluations was a small cluster of 10 nodes
built on the XSEDE Comet cloud computing infrastructure at SDSC (San
Diego Supercomputer Center). Each Comet node has Xeon E5-2680v3 at 2.5GHz,
with 24 cores, 128GB RAM, and 320GB SSD.
For our experiments, we used Apache Spark
2.2.0 running on Apache Hadoop 2.6.0. 
All experiments were done on random data, stored in Spark RDDs.
Each Spark executor was configured to have 4 cores and 23 GB RAM.
Consequently, there were 5 executors per node, giving a total of 50
executors, from which 2 were reserved for other tasks.  Each program
was evaluated over 5 datasets and each evaluation was repeated 4
times, the first of which was discarded to make sure that the JVM/JIT
warm-up time does not skew the results.  Hence, each data point in the
plots in Figure~\ref{performance-results} represents the mean value of
these 3 evaluations. The dataset size was calculated by multiplying
the dataset length by the size of a dataset element when is serialized
to bytes using Java standard serialization.\extended{ For example, a sparse
matrix of type \s{RDD[((Long,Long),Double)]} with 1000 elements has
size $1000*234$ bytes because \s{((Long,Long),Double)} is serialized
to 234 bytes.} Each program was evaluated in 3 different ways: as a
loop-based program translated by DIABLO to the Spark Core API (RDDs)
(the lines tagged ``DIABLO"), as an equivalent efficient program in
Spark Core written by us (the lines tagged ``hand-written"), and as a
loop-based program translated by \textsc{Casper} to Spark Core, when
such translation is possible (the lines tagged ``Casper").  These programs are available on
GitHub~\cite{diql} and are given in Appendix~\ref{benchmark-programs}.

Conditional Sum filters a dataset \s{V} of type \s{RDD[Double]} that
contains random data and aggregates the result. The Spark code is
\lstinline$V.filter(_ < 100).reduce(_+_)$.
The largest dataset used had $10^9$ elements and size 7.45 GB.
Equal, String Match, and Word Count used the same dataset of
type \s{RDD[String]} that contains random strings of size 4 so that
there were 1000 different strings.  The largest dataset used had
$2\times10^8$ elements and size 1.49 GB.  Equal checks whether all
the strings in the dataset are equal.  String Match checks whether
the dataset contains ``key1", ``key2", or ``key3".  For each
different string in the dataset, WordCount counts how many times this
string occurs.  Histogram scans a dataset \s{P} of RGB pixels of
type \s{RDD[(Int,Int,Int)]}, and for each one of the RGB components,
it creates a histogram.  For instance, the Spark code for the red
component is \lstinline$P.map(_.1).countByValue()$.  The largest dataset
used had $2\times 10^8$ elements and size 10.99 GB.  Linear
Regression takes a dataset of 2-D points of
type \s{RDD[(Double,Double)]} and calculates the intercept and the
slope coefficient that models the dataset.  The data used were points
$(x+dx,x-dx)$, where $x$ is a random double between 0 and 1000 and
$dx$ is a random double between 0 and 10.  The largest dataset used
had $2\times 10^8$ elements and size 27.99 GB.  Group By groups a dataset
of type \s{RDD[(Long,Double)]} by it first component and sums up
the second component. The keys were random long integers with 10
duplicates on the average.  The largest dataset used had $2\times
10^8$ elements and size 8.75 GB. We can see that programs generated by
DIABLO have performance comparable to the hand-written programs and are faster than those
by \textsc{Casper}.

{\em Matrix addition and multiplication:}
The sparse matrices used in our
experiments have type \s{RDD[((Long,Long),Double)]}.
Although sparse, all matrix elements were provided,
were placed in random order, and
were filled with random values between 0.0 and 10.0.
\extended{The DIABLO matrix multiplication
program is given in the Introduction,
while the hand-written Spark program is as follows:
\begin{tabbing}
M.\=\+map{ case ((i,j),m) => (j,(i,m)) }\\
 .join( N.map{ case ((i,j),n) => (i,(j,n)) } )\\
 .map{ case (k,((i,m),(j,n))) => ((i,j),m*n) }\\
 .reduceByKey(\_+\_)
\end{tabbing}}
The matrices used for addition and multiplication
were pairs of square matrices of the same size.
The largest matrices used in addition had $8000\times 8000$ elements
and size 13.83 GB each, while those in
multiplication had $4000\times 4000$ elements and size 3.46 GB each.
The results are shown in Figures~\ref{performance-results}.H and I.
We can see that here too programs generated by
DIABLO have performance comparable to the hand-written programs.

\ignore{
The generated DIABLO code too generated the
same operations but it also included an extra operation to add the
initial values of \s{R[i,j]} in \s{R[i,j] += M[i,k]*N[k,j]} (since it
could not tell that the initial values were 0.0) and to incorporate
the values of \s{R[i,j]} that are not changed during the loop (since
it could not tell that the loop replaced all values).
A more sophisticated optimizer could have caught these
optimizations and removed this operation.}

{\em PageRank:} The PageRank program is one iteration of the page-rank
algorithm that assigns a rank to each vertex of a graph, which measures
its importance relative to the other vertices in the graph.  The
graphs used in our experiments were synthetic data generated by the
RMAT (Recursive MATrix) Graph Generator~\cite{Chakrabarti:sdm04} using
the Kronecker graph generator parameters a=0.30, b=0.25, c=0.25, and
d=0.20. The number of edges generated were 10 times the number of
graph vertices. The largest graph used had $2\times 10^7$ vertices, $2\times 10^8$
edges, and had size 36.32 GB. The results are shown in
Figure~\ref{performance-results}.J. The pagerank step in the
hand-written program was simply a join between the graph and the
current pagerank, followed by a reduceByKey. The generated DIABLO
program though used a triple join among the graph, the current pagerank, and
the node fan-out vector, followed by a reduceByKey.

{\em K-Means clustering:}
The KMeans program computes one iteration step of the K-Means clustering
algorithm, which finds the $K$ centroids of a set of 2-D points on a plane.
The datasets used in our experiments are random points on a plane inside a
$10\times 10$ grid of squares, where each square has a top-left corner
at $(i*2+1,j*2+1)$ and bottom-right corner at $(i*2+2,j*2+2)$, for
$i\in [0,9]$ and $j\in [0,9]$. That is, there should be 100
centroids, which are the square centers $(i*2+1.5,j*2+1.5)$. The
initial centroids were set to be the points $(i*2+1.2,j*2+1.2)$. 
The largest dataset used had $10^7$ data points and size 1.36 GB.
The results are shown in Figure~\ref{performance-results}.K. The
hand-written program broadcasts the initial centroids to all workers
so that each worker keeps a copy in its memory, and then uses a map
followed by a reduceByKey, in which the shuffled data were very small
and of constant size. On the other hand, DIABLO stores the centroids
into an RDD and uses Spark joins to correlate points with centroids,
making the entire process expensive.

{\em Matrix factorization:} The last program to evaluate is
one iteration of matrix factorization using gradient
descent~\cite{koren:comp09}. The loop-based program was given
in Section~\ref{recurrences}.
\ignore{ The hand-writen program
that calculates the matrix factors \s{P} and \s{Q} from \s{R} is as follows:
val E = op( _-_, R, multiply(P,Q) )
P = op( _+_, P, op( _-_, 
    multiply(E,transpose(Q)).mapValues(_*2),
    P.mapValues(_*b) ).mapValues(_*a) )
Q = op( _+_, Q, op( _-_,
    transpose(multiply(transpose(E),P)).mapValues(_*2),
    Q.mapValues(_*b) ).mapValues(_*a) )
where \s{multiply} and \s{transpose} are matrix multiplication and transpose, and
\lstinline$op (f,x,y)$ is \lstinline$x.join(y).mapValues(f)$.}
For our experiments, we used the learning rate $a=0.002$ and the
normalization factor $b=0.02$. The matrix to be factorized, \s{R}, was
a square sparse matrix $n*n$ with random integer values between 1 and
5, in which only the 10\% of the elements were provided (the rest were
implicitly zero). The derived matrices \s{P} and \s{Q} had dimensions
$n*2$ and $2*n$, respectively, and they were initialized with random
values between 0.0 and 1.0. The largest matrix $R$ used had
$8000*8000$ elements and size 13.65 GB. The results are
shown in Figure~\ref{performance-results}.L.

From these experiments, we can see that, except KMeans and Matrix
Factorization, the DIABLO generated programs have performance
comparable to the hand-written programs, although the DIABLO PageRank
performance was erratic.  These three programs are far more complex
than the others, which caused DIABLO to generate some unnecessary
joins. These joins could have been eliminated by a more sophisticated
query optimizer.  The focus of our current work is on generating
correct DISC programs from vector-based loops. We are planning to
explore more effective query optimization techniques in a future work.

\section{Conclusion}

We have addressed the problem of automated parallelization of
array-based loops by translating them to comprehensions, which can
then be translated and optimized to distributed data parallel
operations.  The efficiency of our translations would mostly depend on
the effectiveness of code optimization after translation, which we are
planning to address more thoroughly in a future work. We are also
planning to look at cost-based optimizations, such as determining
whether an array is small enough to fit in a worker's memory in order
to broadcast it to all workers, thus speeding up joins over this
array.\extended{ For example, the centroid vector in the K-means clustering
example was small enough to broadcast.} One source of inefficiency in
our translations is the large number of generated joins. When two
arrays are used together in a program, such as in $A[i]*B[i]$, this
term is translated to a join between $A$ and $B$. This join can be
avoided if we co-partition these two vectors using the same
partitioner. Then, $A[i]*B[i]$ can be implemented using the
zipPartitions operation in Spark, which does not cause any shuffling.
As a future work, we are also planning to experiment with more
platforms as the target of DIABLO, such as translating comprehensions
to Spark SQL, which supports cost-based optimization.
\extended{Finally, as a future work, we want to investigate how to generate
optimal parallel algorithms from loops, such as the SUMMA
algorithm~\cite{summa} for distributed matrix multiplication. For such
optimizations, a template-based approach, where a generic template of
algebraic operations is mapped to an algorithm, may be a more suitable
solution.}

\begin{figure*}
\framebox{
\begin{minipage}[l]{6.78in}
\begin{subequations}
\begin{align}
\intertext{\underline{$\mathcal{E}'\sem{e}_\sigma:$\qquad Translate the expression $e$ to a a comprehension term}}
\mathcal{E}'\sem{V}_\sigma & =\;
\left\{\begin{array}{ll}
\sigma(V) & \mbox{if $V\in\sigma$}\\
\some{V} & \mbox{otherwise}
\end{array}\right.\label{u0}\\
\mathcal{E}'\sem{e}_\sigma & =\;\mathcal{E}\sem{e}\qquad\mbox{\em (as defined in Figure~\ref{translation-rules})}
\end{align}
\end{subequations}
\begin{subequations}
\begin{align}
\nonumber\\[-3.2em]
\intertext{\underline{$\mathcal{U}'\sem{d}_\sigma(x):$\qquad Transform the state $\sigma$ by replacing the destination $d$ with the value $x$}}
\mathcal{U}'\sem{V}_\sigma(x) & =\;\sigma[V=x]\label{uu1}\\
\mathcal{U}'\sem{d.A_i}_\sigma(x) & =\;\mathcal{U}'\sem{d}_\sigma(\ocompr{\lt\,A_1\req w.A_1,\ldots,A_i\req v,\ldots,A_n\req w.A_n\,\gt}{w\from\mathcal{E}'\sem{d}_\sigma,\,v\from x})\label{uu2}\\
\mathcal{U}'\sem{V[e]}_\sigma(x) & =\;\sigma[V=\ocompr{V\lhd\bag{(i,v)}}{i\from\mathcal{E}'\sem{e}_\sigma,\,v\from x}]\label{uu3}\\
\mathcal{U}'\sem{V[e_1,e_2]}_\sigma(x) & =\;\sigma[V=\ocompr{V\lhd\bag{((i,j),v)}}{i\from\mathcal{E}'\sem{e_1}_\sigma,\,j\from\mathcal{E}'\sem{e_2}_\sigma,\,v\from x}]\label{uu4}
\end{align}
\end{subequations}
\begin{subequations}
\begin{align}
\nonumber\\[-3.2em]
\intertext{\underline{$\mathcal{T}\sem{s}_\sigma:$\qquad Translate the parallelizable program $s$ to a term that transforms the state $\sigma$}}
\mathcal{T}\sem{d\,\opluseq e}_\sigma & =\; \mathcal{T}\sem{d\,:=d\oplus e}_\sigma\label{tt1}\\
\mathcal{T}\sem{d:=e}_\sigma & =\;\mathcal{U}'\sem{d}_\sigma(\mathcal{E}'\sem{e}_\sigma)\label{tt2}\\
\mathcal{T}\sem{\mathbf{for}\;v=e_1,e_2\;\mathbf{do}\;s}_\sigma
& =\;(\compr{\lm x.\,\mathcal{T}\sem{s}_x}{v_1\from\mathcal{E}'\sem{e_1}_\sigma,\,v_2\from\mathcal{E}'\sem{e_2}_\sigma,\,v\from\mathrm{range}(v_1,v_2)}_\circ)\,\sigma\label{tt4}\\
\mathcal{T}\sem{\mathbf{for}\;v\;\mathbf{in}\;e\;\mathbf{do}\;s}_\sigma
& =\;(\compr{\lm x.\,\mathcal{T}\sem{s}_x}{A\from\mathcal{E}'\sem{e}_\sigma,\,(i,v)\from A}_\circ)\,\sigma\label{tt5}\\
\mathcal{T}\sem{\mathbf{if}\;(e)\;s_1\;\mathbf{else}\;s_2}_\sigma
& =\;\ocompr{x}{p\from\mathcal{E}'\sem{e}_\sigma,\,x\from\mathbf{if}\;p\;\mathbf{then}\;\mathcal{T}\sem{s_1}_\sigma\;\mathbf{else}\;\mathcal{T}\sem{s_2}_\sigma}\label{tt7}
\end{align}
\end{subequations}
\end{minipage}
}
\caption{Semantics of a parallelizable program}\label{default-semantics}
\end{figure*}

\short{\newpage\balance}
\bibliographystyle{abbrv}

\begin{thebibliography}{10}

\bibitem{tensorflow}
M.~Abadi, P.~Barham, J.~Chen, {\em et al}.
\newblock {TensorFlow: a system for large-scale machine learning}.
\newblock In {\em USENIX Conference on Operating Systems Design and Implementation (OSDI)}, pages 265--283, 2016.

\bibitem{casper:sigmod18}
M.~B.~S.~Ahmad and A.~Cheung.
\newblock {Automatically Leveraging MapReduce Frameworks for Data-Intensive Applications}.
\newblock In {\em ACM SIGMOD International Conference on Management of Data}, pages 1205--1220, 2018.

\bibitem{aho:book}
A.~V.~Aho, M.~S.~Lam, R.~Sethi, and J.~D.~Ullman.
\newblock{Compilers: Principles, Techniques, and Tools (2nd Edition)}.
\newblock {\em Chapter 11: Optimizing for Parallelism and Locality}, Addison Wesley, 2007.

\bibitem{flink}
Apache Flink.
Available: \url{http://flink.apache.org/}, 2020.

\bibitem{hadoop}
Apache Hadoop. Available: \url{http://hadoop.apache.org/}, 2020.

\bibitem{spark}
Apache Spark.
Available: \url{http://spark.apache.org/}, 2020.

\bibitem{sparkSQL:sigmod15}
M.~Armbrust, R.~S.~Xin, C.~Lian, Y.~Huai, D.~Liu, J.~K.~Bradley, X.~Meng, T.~Kaftan, M.~J.~Franklin, A.~Ghodsi, and M.~Zaharia.
\newblock {Spark SQL: Relational Data Processing in Spark}.
\newblock In {\em ACM SIGMOD International Conference on Management of Data}, pages 1383--1394, 2015.

\bibitem{nesl}
G. E. Blelloch and G. W. Sabot.
\newblock{Compiling collection-oriented languages onto massively parallel computers}.
\newblock In {\em Journal of Parallel and Distributed Computing (JPDC)}, 8:119--134, 1990.

\extended{
\bibitem{scihadoop:sc11}
J. Buck, N. Watkins, J. Lefevre, K. Ioannidou, C. Maltzahn, N. Polyzotis, and S. A. Brandt.
\newblock {SciHadoop: Array-based Query Processing in Hadoop}.
\newblock In {\em International Conference for High Performance Computing, Networking, Storage and Analysis (SC)}, 2011.
}

\bibitem{casper:site}
Casper.
Available: \url{http://casper.uwplse.org/}, accessed in January 2020.

\bibitem{Chakrabarti:sdm04}
D.~Chakrabarti, Y.~Zhan, and C.~Faloutsos.
\newblock {R-MAT: A Recursive Model for Graph Mining}. 
\newblock In {\em SIAM International Conference on Data Mining (SDM)}, pages 442--446, 2004.

\bibitem{dean:osdi04}
J.~Dean and S.~Ghemawat.
\newblock {MapReduce: Simplified Data Processing on Large Clusters}.
\newblock In {\em Symposium on Operating Systems Design and Implementation (OSDI)}, 2004.

\bibitem{diql}
DIQL: A Data Intensive Query Language.
Available: \url{https://github.com/fegaras/DIQL}, 2020.

\bibitem{emani:sigmod16}
K.~V.~Emani, K.~Ramachandra, S.~Bhattacharya, and S.~Sudarshan.
\newblock {Extracting Equivalent SQL from Imperative Code in Database Applications}.
\newblock In {\em ACM SIGMOD International Conference on Management of Data}, pages 1781--1796, 2016.

\bibitem{fan:sigmod17}
W.~Fan, J.~Xu, Y.~Wu, W.~Yu, J.~Jiang, Z.~Zheng, B.~Zhang, Y.~Cao, and C.~Tian.
\newblock {Parallelizing Sequential Graph Computations}.
\newblock In {\em ACM SIGMOD International Conference on Management of Data}, pages 495--510, 2017.

\extended{
\bibitem{farzan:pldi17}
A. Farzan and V. Nicolet.
\newblock {Synthesis of Divide and Conquer Parallelism for Loops}.
\newblock In {\em ACM SIGPLAN Conference on Programming Language Design and Implementation (PLDI)}, pages 540–-555. 2017.

\bibitem{farzan:pldi19}
A. Farzan and V. Nicolet.
\newblock {Modular Synthesis of Divide-and-Conquer Parallelism for Nested Loops}.
\newblock In {\em ACM SIGPLAN Conference on Programming Language Design and Implementation (PLDI)}, 2020.
}

\extended{
\bibitem{fedyukovich:pldi17}
G.~Fedyukovich, M.~B.~.S.~Ahmad, and R.~Bodik.
\newblock {Gradual Synthesis for Static Parallelization of Single-pass Array-processing Programs},
\newblock In {\em ACM SIGPLAN Conference on Programming Language Design and Implementation (PLDI)}, 2017.
}

\extended{
\bibitem{dexa16a}
L. Fegaras.
\newblock {A Query Processing Framework for Array-Based Computations}.
\newblock In {\em 27th International Conference on Database and Expert Systems Applications (DEXA)}, 2016.
}

\bibitem{jfp17}
L.~Fegaras.
\newblock {An Algebra for Distributed Big Data Analytics}.
\newblock {\em Journal of Functional Programming}, special issue on Programming Languages for Big Data, Volume 27, 2017.
\ignore{
\bibitem{sigmod95}
L.~Fegaras and D.~Maier.
\newblock {Towards an Effective Calculus for Object Query Languages}.
\newblock In {\em International Conference on Management of Data (SIGMOD)}, pages 47--58, 1995.
}
\bibitem{tods00}
L.~Fegaras and D.~Maier.
\newblock {Optimizing Object Queries Using an Effective Calculus}.
\newblock {\em ACM Transactions on Database Systems (TODS)}, 25(4):457--516, 2000.

\bibitem{diql:BigData}
L.~Fegaras and M.~H.~Noor.
\newblock {Compile-Time Code Generation for Embedded Data-Intensive Query Languages}.
\newblock In {\em IEEE BigData Congress}, 2018.

\short{
\bibitem{diablo-extended}
L.~Fegaras and M.~H.~Noor.
\newblock {Translation of Array-Based Loops to Distributed Data-Parallel Programs (extended paper)}.
Available: \url{http://lambda.uta.edu/drafts/diablo.pdf}.
}

\bibitem{fisher:pldi94}
A.~L.~Fisher and A.~M.~Ghuloum.
\newblock {Parallelizing Complex Scans and Reductions}.
\newblock {\em ACM SIGPLAN Notices}, 29(6):135-–146, 1994.
%
\extended{
\bibitem{summa}
R. A. Geijn and J.~Watts.
\newblock {SUMMA: Scalable Universal Matrix Multiplication Algorithm}.
\newblock In {\em Concurrency: Practice and Experience}, 9(4):255--274, April 1997.

\bibitem{scihive}
Y. Geng, X. Huang, M. Zhu, H. Ruan, and G. Yang.
\newblock {SciHive: Array-based query processing with HiveQL}.
\newblock In {\em IEEE International Conference on Trust, Security and Privacy in Computing and Communications (Trustcom)}, 2013.

\bibitem{systemML}
A. Ghoting, R. Krishnamurthy, E. Pednault, B. Reinwald, V. Sindhwani, S. Tatikonda, Y. Tian, and S. Vaithyanathan.
\newblock {SystemML: Declarative Machine Learning on MapReduce}.
\newblock In {\em IEEE International Conference on Data Engineering (ICDE)}, 2011.
}

\bibitem{guravannavar:vldb08}
R. Guravannavar and S. Sudarshan.
\newblock {Rewriting Procedures for Batched Bindings}.
\newblock {\em PVLDB}, 1(1):1107--1123, 2008.

\bibitem{doall-kavi}
A.~R.~Hurson, J.~T.~Lim, K.~M.~Kavi, and B.~Lee.
\newblock {Parallelization of DOALL and DOACROSS Loops -- a Survey}.
\newblock {\em Advances in Computers}, vol 45, pages 53--103, 1997.

\bibitem{jiang:pact18}
P.~Jiang, L.~Chen, and G.~Agrawal.
\newblock {Revealing Parallel Scans and Reductions in Recurrences through Function Reconstruction}.
\newblock In {\em International Conference on Parallel Architectures and Compilation Techniques (PACT)}, pages 1–-13, 2018.

\bibitem{koren:comp09}
Y. Koren, R. Bell, and C. Volinsky.
\newblock {Matrix Factorization Techniques for Recommender Systems}
\newblock {\em IEEE Computer}, 42(8):30–-37, August 2009.

\extended{
\bibitem{MLbase}
T. Kraska, A. Talwalkar, J. Duchi, R. Griffith, M. Franklin, and M.I. Jordan.
\newblock {MLbase: A Distributed Machine Learning System}.
\newblock In {\em Conference on Innovative Data Systems Research}, 2013.
}

\extended{
\bibitem{lara}
A. Kunft, A. Katsifodimos, S. Schelter, S. Bre\ss, T. Rabl, and V. Markl.
\newblock {An Intermediate Representation for Optimizing Machine Learning Pipelines}.
\newblock {\em PVLDB}, 12(11):1553-1567, 2020.
}

\extended{
\bibitem{MLlib:mlr16}
X. Meng, J. Bradley, B. Yavuz, {\em et al}.
\newblock{MLlib: Machine Learning in Apache Spark}.
\newblock In {\em Journal of Machine Learning Research}, 17:1-7, 2016.
}

\bibitem{morita:pldi07}
K.~Morita, A.~Morihata, K.~Matsuzaki, Z.~Hu, and M.~Takeichi.
\newblock {Automatic inversion generates divide-and-conquer parallel programs}.
\newblock In {\em ACM SIGPLAN Conference on Programming Language Design and Implementation (PLDI)}, pages 146--155, 2007.

\bibitem{palmer:95}
D. W. Palmer, J. F. Prins, and S. Westfold.
\newblock {Work-Efficient Nested Data-Parallelism}.
\newblock In {\em Symposium on the Frontiers of Massively Parallel Processing}, 1995.

\bibitem{tiledb}
S. Papadopoulos, K. Datta, S. Madden, and T. Mattson.
\newblock {The TileDB array data storage manager}.
\newblock {\em PVLDB}, 10(4):349–360, 2016.

\bibitem{mold:oopsla14}
C.~Radoi, S.~J.~Fink, R.~Rabbah, and M.~Sridharan.
\newblock {Translating Imperative Code to MapReduce}.
\newblock In {\em ACM International Conference on Object Oriented Programming Systems Languages \& Applications (OOPSLA)}, pages 909–-927, 2014.

\bibitem{scidb:sigmod10}
The SciDB Development Team.
\newblock {Overview of SciDB: Large Scale Array Storage, Processing and Analysis}.
\newblock In {\em ACM SIGMOD International Conference on Management of Data}, pages 963-–968, 2010.

\bibitem{smith:pldi16}
C. Smith and A. Albarghouthi.
\newblock {MapReduce Program Synthesis}.
\newblock In {\em ACM SIGPLAN Conference on Programming Language Design and Implementation (PLDI)}, pages 326--340, 2016.

\bibitem{arraystore:sigmod11}
E. Soroush, M. Balazinska, and D. Wang.
\newblock {ArrayStore: A Storage Manager for Complex Parallel Array Processing}.
\newblock In {\em ACM SIGMOD International Conference on Management of Data}, pages 253–-264, 2011.

\extended{
\bibitem{scidb:ssdbm15}
E. Soroush, M. Balazinska, S. Krughoff, and A. Connolly.
\newblock {Efficient Iterative Processing in the SciDB Parallel Array Engine}.
\newblock In {\em 27th International Conference on Scientific and Statistical Database Management (SSDBM)}, 2015.
}

\extended{
\bibitem{thusoo:pvldb09}
A.~Thusoo, J.~S.~Sarma, N.~Jain, Z.~Shao, P.~Chakka, S.~Antony, H.~Liu, P.~Wyckoff, and R.~Murthy.
\newblock {Hive: a Warehousing Solution over a Map-Reduce Framework}.
\newblock {\em PVLDB}, 2(2):1626--1629, 2009.
}

\bibitem{venkat:sc16}
A.~Venkat, M.~S.~Mohammadi, J.~Park, H.~Rong, R.~Barik, M.~M.~Strout, and M.~Hall.
\newblock {Automating Wavefront Parallelization for Sparse Matrix Computations}.
\newblock In {\em International Conference for High Performance Computing,
Networking, Storage and Analysis (SC)}, Article No. 41, pages 1–-12, 2016.

\extended{
\bibitem{scimate}
Y. Wang, W. Jiang, and G. Agrawal.
\newblock {SciMATE: A novel MapReduce-like framework for multiple scientific data formats}.
\newblock In {\em IEEE/ACM International Symposium on Cluster, Cloud and Grid Computing (CCGrid)}, 2012.
}

\extended{
\bibitem{DryadLINQ}
Y.~Yu, M.~Isard, D.~Fetterly, M.~Budiu, U.~Erlingsson, P.~K.~Gunda, and J.~Currey.
\newblock{DryadLINQ: A System for General-Purpose Distributed Data-Parallel Computing Using a High-Level Language}.
\newblock In {\em Symposium on Operating Systems Design and Implementation (OSDI)}, 2008.
}

\bibitem{spark:nsdi12}
M.~Zaharia, M.~Chowdhury, T.~Das, A.~Dave, J.~Ma, M.~McCauley, M.~J.~Franklin, S.~Shenker, and I.~Stoica.
\newblock {Resilient Distributed Datasets: A Fault-Tolerant Abstraction for In-Memory Cluster Computing}.
\newblock In {\em USENIX Symposium on Networked Systems Design and Implementation (NSDI)}, 2012.

\end{thebibliography}

\appendix

\section{Correctness Proofs}\label{proofs}

\begin{customthm}{\ref{for-distribution}}
An affine for-loop satisfies:
\begin{align}
&\textbf{for}\;i=\ldots\;\textbf{do}\;\{\,s_1;\,s_2\,\}\nonumber\\
&\hspace*{3ex}\;=\;\{\,\textbf{for}\;i=\ldots\;\textbf{do}\;s_1;\;\textbf{for}\;i=\ldots\;\textbf{do}\;s_2\;\}\tag{\ref{for-distr}}
\end{align}
\end{customthm}
\begin{proof}
Based on the exception~(a) in Definition~\ref{PFOR-def}, if there is
$d\in\mathcal{W}\sem{s_1}$ (ie, there is an update $d:=e$ in $s_1$)
and $d\in\mathcal{R}\sem{s_2}$ (ie, $d$ is read in $s_2$), then
according to restriction~(1) in Definition~\ref{PFOR-def}, $d$ must be
affine. That is, the location of $d$ is different for different values
of the loop index $i$, which means that there is no interference
across iteration steps.  Hence, we can do all the updates $d:=e$ first
in one loop and then read all $d$ values in a separate loop. Based on
the exception~(b) in Definition~\ref{PFOR-def}, if there is $d\opluseq e$
in $s_1$ and $d\in\mathcal{R}\sem{s_2}$, then $\mathrm{affine}(d,s_2)$ and
$\mathrm{context}(s_1)\cap\mathrm{context}(s_2)=\mathrm{indexes}(d)$.
That is, $i\in\mathrm{indexes}(d)$ since both statements are inside
the for-loop of $i$.  Furthermore, $d$ must be affine in the context of
$s_2$, which contains $i$.  Hence, like the previous case for an
update to $d$, we can calculate all increments $d\opluseq e$ first in
one loop and then read all $d$ values in a separate loop.  For all
other cases, based on restriction~(2) in Definition~\ref{PFOR-def},
there are no interferences between $s_1$ and $s_2$ and,
therefore, the loop can be split into two loops.
\end{proof}

\begin{theorem}[Soundness]\label{correctness-proof}
The transformation rules in Figure~\ref{translation-rules} are meaning preserving
for all programs that satisfy the recurrence restrictions in Definition~\ref{PFOR-def}.
\end{theorem}
\begin{proof}
If we use Equation~(\ref{for-distr}) in Theorem~\ref{for-distribution}
as a rewrite rule from left to right, then any loop-based program can
be put to a normal form that consists of a single sequential program
$ss$ that contains parallelizable for-loops $ps$:
\[\begin{array}{rcll}
\multicolumn{4}{l}{\mbox{\textbf{Sequential Program}:}}\\
ss &::=&\mathbf{var}\;v:t=e & \mbox{declaration}\\
&|&\mathbf{while}\;(e)\;ss & \mbox{loop}\\
&|&\{\,ss_1;\,\ldots;\,ss_n\} & \mbox{statement block}\\
&|& ps\\[1ex]
\multicolumn{4}{l}{\mbox{\textbf{Parallelizable Program}:}}\\
ps & ::= & d\opluseq e & \mbox{incremental update}\\
&|& d:=e & \mbox{assignment}\\
&|&\mathbf{for}\;v=e_1,e_2\;\mathbf{do}\;ps & \mbox{iteration}\\
&|&\mathbf{for}\;v\;\mathbf{in}\;e\;\mathbf{do}\;ps & \mbox{traversal}\\
&|&\mathbf{if}\;(e)\;ps_1\;\mathbf{else}\;ps_2 & \mbox{conditional}
\end{array}\]
The sequential part generated by the rules in
Figure~\ref{translation-rules} is equivalent to the derived sequential
program $ps$. Therefore, to prove Theorem~\ref{correctness-proof}, we
need to prove that the derived parallelizable programs are equivalent
to those generated by the rules in Figure~\ref{translation-rules}. To
prove this equivalence, we have to give formal semantics to the
parallelizable programs in the form of monoid comprehensions and then
prove the equivalence of the derived monoid comprehensions. The formal semantics of a
parallelizable program is given by the rules in
Figure~\ref{default-semantics}. It is based on denotational
semantics, which is often used to ascribe a formal meaning to imperative
programming languages. In denotational semantics, to capture the
meaning of an imperative program, we use a state $\sigma$ to
encapsulate all updatable variables. Then, a statement is translated
to a state transformer, which is a function from the current state to
a new state. If the statement does not cause any side effects, the
state is propagated as is; otherwise, the state is replaced with a new
state that reflects the updates. In our semantics, the state $\sigma$ is a map
from a variable name to a bag of values, where the bag can have 0 or 1 elements.
Then, $\sigma(V)$ returns this bag and $\sigma[V=v]$ replaces the value of the
variable $V$ with the bag $v$.

The Rules~(\ref{tt1}) and~(\ref{tt2}) in
Figure~\ref{default-semantics} translate assignments and incremental updates to
state transformers using Rules~(\ref{uu1})-(\ref{uu4}), which replace
the $\sigma$ component associated with the destination $d$ with a new
value $x$. To capture the sequential semantics of a for-loop, we use
a comprehension $\compr{f}{\bar{q}}_\circ$ over the monoid
$\circ$. The function composition $\circ$ satisfies
$(f_2\circ f_1)(x)=f_2(f_1(x))$, that is, it combines two state
transformers. Hence, the comprehension $\compr{f}{\bar{q}}_\circ$
returns a function, which when applied to a state $\sigma$, it
transforms $\sigma$ to $f(f(\ldots f(\sigma)))$. This means that, if
$f$ is the state transformer associated with the body of a for-loop, then this
comprehension captures the for-loop iteration.

We first prove the following theorem using structural induction on the statement $s$:
\begin{align}
\Sigma_\sigma(\mathcal{S}\sem{s}(\bar{q}))=(\compr{\lm x.\,\mathcal{T}\sem{s}_x}{\bar{q}}_\circ)\,\sigma\label{st-law}
\end{align}
where
$\Sigma_\sigma(\liste{V_1:=v_1,\ldots,V_n:=v_n})=\sigma[V_1=v_1,\ldots,$ $V_n=v_n]$,
that is, $\Sigma_\sigma$ converts a list of updates to a state
transformation. For the induction base case, we will only prove
Equation~(\ref{st-law}) for the statement $V[e_1]\opluseq e_2$. (The
other cases are easier to prove.)  The left side is:
\[\begin{array}{l}
\Sigma_\sigma(\mathcal{S}\sem{V[e_1]\opluseq e_2}(\bar{q}))\\
\why{from Equation~(\ref{s1})}
=\Sigma_\sigma(\mathcal{U}\sem{V[e_1]}(\compr{(k,w\oplus(\oplus/v))}{\bar{q},\,v\from\mathcal{E}\sem{e_2},\\
\hspace*{3ex}k\from\mathcal{K}\sem{V[e_1]},\;\mathbf{group\ by}\,k,\,w\from\mathcal{D}\sem{V[e_1]}(k)}))\\
\why{from Equation~(\ref{u12})}
=\sigma[V=\some{\sigma.V\lhd\compr{(k,w\oplus(\oplus/v))}{\bar{q},\,v\from\mathcal{E}\sem{e_2},\\
\hspace*{3ex}k\from\mathcal{K}\sem{V[e_1]},\;\mathbf{group\ by}\,k,\,w\from\mathcal{D}\sem{V[e_1]}(k)}}]\\
\why{from Equations~(\ref{k4}) and~(\ref{d4})}
=\sigma[V=\some{\sigma.V\lhd\compr{(k,w\oplus(\oplus/v))}{\bar{q},\,v\from\mathcal{E}\sem{e_2},\\
\hspace*{3ex} k\from\mathcal{E}\sem{e_1},\;\mathbf{group\ by}\,k,\,(i,w)\from\sigma.V,\,i=k}}]\\
=\sigma[V=\some{\sigma.V\lhd\compr{(k,w\oplus(\oplus/s))}{(k,s)\from G,\\
\hspace*{20ex}(i,w)\from\sigma.V,\,i=k}}]\\
\why{defined as:}
=M(\sigma,G)
\end{array}\]
where the key-value map $G$ is:
\begin{align*}
G &=\compr{(k,v)}{\bar{q},\,v\from\mathcal{E}\sem{e_2},\,k\from\mathcal{E}\sem{e_1},\;\mathbf{group\ by}\,k}
\end{align*}
The right side of Equation~(\ref{st-law}) is:
\[\begin{array}{l}
(\compr{\lm x.\,\mathcal{T}\sem{V[e_1]\opluseq e_2}_x}{\bar{q}}_\circ)\,\sigma\\
\hspace*{6ex}\mbox{{\em (from Equations~(\ref{tt1}) and~(\ref{tt2})}}\\
\hspace*{6ex}\mbox{{\em \ \ and after normalization)}}\\
=(\compr{\lm x.\,\mathcal{U}'\sem{V[e_1]}_x(\mathcal{E}'\sem{V[e_1]\oplus e_2}_x)}{\bar{q}}_\circ)\,\sigma\\
\why{from Equations~(\ref{uu3}) and~(\ref{t4})}
=(\compr{\lm x.\,x[V=x.V\lhd\ocompr{(j,w\oplus v)}{j\from\mathcal{E}'\sem{e_1}_x,\\
\hspace*{8ex}k\from\mathcal{E}'\sem{e_1}_x,\,(i,w)\from x.V,\,i=k,\\
\hspace*{8ex}v\from\mathcal{E}'\sem{e_2}_x}]}{\bar{q}}_\circ)\,\sigma\\
\hspace*{6ex}\mbox{{\em (after removing the repeated generator}}\\
\hspace*{6ex}\mbox{{\em \ \ and moving the last generator)}}\\
=(\compr{\lm x.\,x[V=x.V\lhd\ocompr{(k,w\oplus v)}{v\from\mathcal{E}'\sem{e_2}_x,\\
\hspace*{3ex}k\from\mathcal{E}'\sem{e_1}_x,\,(i,w)\from x.V,\,i=k}]}{\bar{q}}_\circ)\,\sigma\\
\hspace*{6ex}\mbox{{\em ($\mathcal{E}'\sem{e_1}_x=\mathcal{E}'\sem{e_1}_\sigma$ and $\mathcal{E}'\sem{e_2}_x=\mathcal{E}'\sem{e_2}_\sigma$}}\\
\hspace*{6ex}\mbox{{\em \ \ because $e_1$ and $e_2$ do not interfere with $V$)}}\\
=(\compr{\lm x.\,x[V=x.V\lhd\ocompr{(k,w\oplus v)}{v\from\mathcal{E}'\sem{e_2}_\sigma,\\
\hspace*{3ex}k\from\mathcal{E}'\sem{e_1}_\sigma,\,(i,w)\from x.V,\,i=k}]}{\bar{q}}_\circ)\,\sigma\\
=(\compr{\lm x.\,x[V=x.V\lhd\ocompr{(k,w\oplus v)}{(i,w)\from x.V,\,i=k}]}\\
\skiptext{=(}{\bar{q},\,v\from\mathcal{E}'\sem{e_2}_\sigma,\,k\from\mathcal{E}'\sem{e_1}_\sigma}_\circ)\,\sigma\\
\why{by unnesting the group-by $G$}
=(\compr{\lm x.\,x[V=x.V\lhd\ocompr{(k,w\oplus v)}{(i,w)\from x.V,\,i=k}]\\
\skiptext{=(}}{(k,s)\from G,\,v\from s}_\circ)\,\sigma\\
\why{defined as:}
=N(\sigma,G)
\end{array}\]
To prove that $M(\sigma,G)=N(\sigma,G)$,
we use induction over $G$. It is easy to prove this equality for an empty and a unary $G$.
For $G=G_1\gb_\uplus G_2$, we assume the equality is true for $G_1$ and $G_2$ (induction hypotheses)
and prove it for $G$ (induction step):
\[\begin{array}{l}
M(\sigma,G)=M(\sigma,G_1\gb_\uplus G_2)\\
=\sigma[V=\some{\sigma.V\lhd\compr{(k,w\oplus(\oplus/s))}{(k,s)\from (G_1\gb_\uplus G_2),\\
\hspace*{20ex}(i,w)\from\sigma.V,\,i=k}}]\\
=\sigma[V=\some{\sigma.V\lhd\compr{(k,w\oplus(\oplus/s))}{(k,s)\from G_1,\\
\hspace*{20ex}(i,w)\from\sigma.V,\,i=k}\\
\skiptext{=\sigma[V=\{\sigma.V}\lhd\compr{(k,w\oplus(\oplus/s))}{(k,s)\from G_2,\\
\hspace*{20ex}(i,w)\from\sigma.V,\,i=k}}]\\
=(\lm x.\,M(x,G_2))(M(\sigma,G_1))\\
\why{induction hypotheses}
=(\lm x.\,N(x,G_2))(N(\sigma,G_1))\\
=N(\sigma,G_1\gb_\uplus G_2)=N(\sigma,G)
\end{array}\]

We will now use the following law:
\begin{align}
\compr{\compr{f}{\bar{q_1}}_\circ}{\bar{q_2}}_\circ = \compr{f}{\bar{q_2},\,\bar{q_1}}_\circ\label{circ-law}
\end{align}
to prove Equation~(\ref{st-law}) for a statement $\mathbf{for}\;v=e_1,e_2\;\mathbf{do}\;s$ (induction step),
assuming that it is true for $s$ (induction hypothesis):
\[\begin{array}{l}
\Sigma_\sigma(\mathcal{S}\sem{\mathbf{for}\;v=e_1,e_2\;\mathbf{do}\;s}(\bar{q}))\\
\why{from Equation~(\ref{x19})}
=\Sigma_\sigma(\mathcal{S}\sem{s}(\bar{q}\app\liste{v_1\from\mathcal{E}\sem{e_1},\,v_2\from\mathcal{E}\sem{e_2},\\
\skiptext{=\Sigma_\sigma(\mathcal{S}\sem{s}(\bar{q}\app [\,}v\from\mathrm{range}(v_1,v_2)}))\\
\why{induction hypothesis}
=\compr{\lm x.\,\mathcal{T}\sem{s}_x}{\bar{q},\,v_1\from\mathcal{E}\sem{e_1},\,v_2\from\mathcal{E}\sem{e_2},\\
\skiptext{=\comprr{\lm x.\,\mathcal{T}\sem{s}_x}}v\from\mathrm{range}(v_1,v_2)}_\circ\,\sigma\\
\why{from Equation~(\ref{circ-law})}
=(\compr{\lm x.\,(\compr{\lm z.\,\mathcal{T}\sem{s}_z}{v_1\from\mathcal{E}'\sem{e_1}_x,\,v_2\from\mathcal{E}'\sem{e_2}_x,\\
\skiptext{a\comprr{\lm x.\,(\comprr{\lm z.\,\mathcal{T}\sem{s}_z}}}v\from\mathrm{range}(v_1,v_2)}_\circ)\,x}{\bar{q}}_\circ)\,\sigma\\
\why{from Equation~(\ref{tt4})}
=(\compr{\lm x.\,\mathcal{T}\sem{\mathbf{for}\;v=e_1,e_2\;\mathbf{do}\;s}_x}{\bar{q}}_\circ)\,\sigma
\end{array}\]
There is a similar proof for the other cases.

The correctness of Theorem~\ref{correctness-proof} comes directly from Equation~(\ref{st-law}) when
$\bar{q}$ is empty:
\[\Sigma_\sigma(\mathcal{S}\sem{s}(\liste{\,}))=(\compr{\lm x.\,\mathcal{T}\sem{s}_x}{}_\circ)\,\sigma
=\mathcal{T}\sem{s}_\sigma\]
which correctly captures the meaning of $s$.
\end{proof}

\section{Benchmark Programs}\label{benchmark-programs}

\renewcommand{\paragraph}[1]{\noindent {\bf #1:}}

\paragraph{Conditional Sum in Spark}
\begin{lstlisting}[language=scala]
V.filter( _ < 100).reduce(_+_)
\end{lstlisting}

\paragraph{Conditional Sum in DIABLO}
\begin{lstlisting}[language=java]
var sum: Double = 0.0;

for v in V do
    if (v < 100)
       sum += v;
\end{lstlisting}

\paragraph{Equal in Spark}
\begin{lstlisting}[language=scala]
val x = V.first()
V.map(_ == x).reduce(_&&_)
\end{lstlisting}

\paragraph{Equal in DIABLO}
\begin{lstlisting}[language=java]
var eq: Boolean = true;

for v in V do
    eq := eq && v == x;
\end{lstlisting}

\paragraph{String Match in Spark}
\begin{lstlisting}[language=scala]
val key1 = "key1"
val key2 = "key2"
val key3 = "key3"
words.map{ w => (w == key1)
                || (w == key2)
                || (w == key3) }
     .reduce(_||_)
\end{lstlisting}

\paragraph{String Match in DIABLO}
\begin{lstlisting}[language=java]
var c: Boolean = false;

for w in words do
    c := c || (w == key1 || w == key2
           || w == key3);
\end{lstlisting}

\paragraph{Word Count in Spark}
\begin{lstlisting}[language=scala]
words.map((_,1)).reduceByKey(_+_)
\end{lstlisting}

\paragraph{Word Count in DIABLO}
\begin{lstlisting}[language=java]
var C: map[String,Int] = map();

for w in words do
    C[w] += 1;
\end{lstlisting}

\paragraph{Histogram in Spark}
\begin{lstlisting}[language=scala]
case class Color ( red: Int, green: Int, blue: Int )
val R = P.map(_.red).countByValue()
val G = P.map(_.green).countByValue()
val B = P.map(_.blue).countByValue()
\end{lstlisting}

\paragraph{Histogram in DIABLO}
\begin{lstlisting}[language=java]
var R: map[Int,Int] = map();
var G: map[Int,Int] = map();
var B: map[Int,Int] = map();

for p in P do {
    R[p.red] += 1;
    G[p.green] += 1;
    B[p.blue] += 1;
};
\end{lstlisting}

\paragraph{Linear Regression in Spark}
\begin{lstlisting}[language=scala]
val x_bar = P.map(_._1).reduce(_+_)/n
val y_bar = P.map(_._2).reduce(_+_)/n

val xx_bar = P.map(x => (x._1-x_bar)*(x._1-x_bar))
              .reduce(_+_)
val yy_bar = P.map(y => (y._2-y_bar)*(y._2-y_bar))
              .reduce(_+_)
val xy_bar = P.map(p => (p._1-x_bar)*(p._2-y_bar))
              .reduce(_+_)
val slope = xy_bar/xx_bar
val intercept = y_bar - slope * x_bar
\end{lstlisting}

\paragraph{Linear Regression in DIABLO}
\begin{lstlisting}[language=java]
var sum_x: Double = 0.0;
var sum_y: Double = 0.0;
var x_bar: Double = 0.0;
var y_bar: Double = 0.0;
var xx_bar: Double = 0.0;
var yy_bar: Double = 0.0;
var xy_bar: Double = 0.0;
var slope: Double = 0.0;
var intercept: Double = 0.0;

for p in P do {
    sum_x += p._1;
    sum_y += p._2;
};

x_bar := sum_x/n;
y_bar := sum_y/n;

for p in P do {
    xx_bar += (p._1-x_bar)*(p._1-x_bar);
    yy_bar += (p._2-y_bar)*(p._2-y_bar);
    xy_bar += (p._1-x_bar)*(p._2-y_bar);
};

slope := xy_bar/xx_bar;
intercept := y_bar-slope*x_bar;
\end{lstlisting}

\paragraph{Group-by in Spark}
\begin{lstlisting}[language=scala]
case class GB ( K: Long, A: Double )
V.map{ case GB(k,v) => (k,v) }.reduceByKey(_+_)
\end{lstlisting}

\paragraph{Group-by in DIABLO}
\begin{lstlisting}[language=java]
var C: vector[Double] = vector();

for v in V do
    C[v.K] += v.A;
\end{lstlisting}

\paragraph{Matrix Addition in Spark}
\begin{lstlisting}[language=scala]
M.join(N).mapValues{ case (m,n) => n + m }
\end{lstlisting}

\paragraph{Matrix Addition in DIABLO}
\begin{lstlisting}[language=java]
var R: matrix[Double] = matrix();

for i = 0, n-1 do
    for j = 0, mm-1 do
        R[i,j] := M[i,j]+N[i,j];
\end{lstlisting}

\paragraph{Matrix Multiplication in Spark}
\begin{lstlisting}[language=scala]
M.map{ case ((i,j),m) => (j,(i,m)) }
 .join( N.map{ case ((i,j),n) => (i,(j,n)) } )
 .map{ case (k,((i,m),(j,n))) => ((i,j),m*n) }
 .reduceByKey(_+_)
\end{lstlisting}

\paragraph{Matrix Multiplication in DIABLO}
\begin{lstlisting}[language=java]
var R: matrix[Double] = matrix();

for i = 0, n-1 do
    for j = 0, n-1 do {
        R[i,j] := 0.0;
        for k = 0, mm-1 do
            R[i,j] += M[i,k]*N[k,j];
};
\end{lstlisting}

\paragraph{PageRank in Spark}
\begin{lstlisting}[language=scala]
val links = E.map(_._1).groupByKey().cache()
var ranks = links.mapValues(v => 1.0/vertices)

for (i <- 1 to num_steps) {
    val contribs
      = links.join(ranks).values.flatMap {
             case (urls, rank)
               => val size = urls.size
                  urls.map(url => (url, rank / size))
        }
    ranks = contribs.reduceByKey(_ + _)
                .mapValues(0.15/vertices + 0.85 * _)
}
\end{lstlisting}

\paragraph{PageRank in DIABLO}
\begin{lstlisting}[language=java]
var P: vector[Double] = vector();
var C: vector[Int] = vector();
var N: Int = vertices;
var b: Double = 0.85;

for i = 1, N do {
    C[i] := 0;
    P[i] := 1.0/N;
};

for i = 1, N do
    for j = 1, N do
       if (E[i,j])
          C[i] += 1;

var k: Int = 0;

while (k < num_steps) {
  var Q: matrix[Double] = matrix();
  k += 1;
  for i = 1, N do
    for j = 1, N do
        if (E[i,j])
           Q[i,j] := P[i];
  for i = 1, N do
      P[i] := (1-b)/N;
  for i = 1, N do
      for j = 1, N do
          P[i] += b*Q[j,i]/C[j];
};
\end{lstlisting}

\paragraph{KMeans Clustering in Spark}
\begin{lstlisting}[language=scala]
var centroids = initial_centroids

def distance ( x, y )
  = Math.sqrt((x._1-y._1)*(x._1-y._1)
              +(x._2-y._2)*(x._2-y._2))

case class Avg ( sum: (Double,Double), count: Long ) {
  def ^^ ( x: Avg ): Avg
    = Avg((sum._1+x.sum._1,sum._2+x.sum._2),count+x.count)
  def value(): (Double,Double)
    = (sum._1/count,sum._2/count)
}

case class ArgMin ( index: Long, distance: Double ) {
  def ^ ( x: ArgMin ): ArgMin
    = if (distance <= x.distance) this else x
}

for ( i <- 1 to num_steps )
   centroids
     = points.map { p => (centroids.minBy(distance(p,_)),
                          Avg(p,1)) }
             .reduceByKey(_ ^^ _)
             .map(_._2.value())
             .collect()
\end{lstlisting}

\paragraph{KMeans Clustering in DIABLO}
\sf{C} and \s{Avg} are defined as Scala Arrays.
\begin{lstlisting}[language=java]
var closest: vector[ArgMin] = vector();

var steps: Int = 0;
while (steps < num_steps) {
   steps += 1;
   for i = 0, N-1 do {
       closest[i] := ArgMin(0,10000.0);
       for j = 0, K-1 do
           closest[i]
              := closest[i]
                 ^ ArgMin(j,distance(P[i],C[j]));
       avg[closest[i].index]
          := avg[closest[i].index] ^^ Avg(P[i],1);
   };
   for i = 0, K-1 do
       C[i] := avg[i].value();
};
\end{lstlisting}

\paragraph{Matrix Factorization in Spark}
\begin{lstlisting}[language=scala]
def transpose ( x )
  = x.map{ case ((i,j),v) => ((j,i),v) }

def op ( f: (Double,Double) => Double, x, y )
  = x.join(y).mapValues{ case ((a,b)) => f(a,b) }

def multiply ( x, y )
  = x.map{ case ((i,j),m) => (j,(i,m)) }
     .join( y.map{ case ((i,j),n) => (i,(j,n)) } )
     .map{ case (k,((i,m),(j,n))) => ((i,j),m*n) }
     .reduceByKey(_+_)

for ( i <- 1 to num_steps ) {
  val E = op( _-_, R, multiply(P,Q) ).cache()
  P = op( _+_, P, op( _-_, 
      multiply(E,transpose(Q)).mapValues(_*2),
      P.mapValues(_*b) ).mapValues(_*a) ).cache()
  Q = op( _+_, Q, op( _-_,
      transpose(multiply(transpose(E),P)).mapValues(_*2),
      Q.mapValues(_*b) ).mapValues(_*a) ).cache()
}
\end{lstlisting}

\paragraph{Matrix Factorization in DIABLO}
\begin{lstlisting}[language=java]
var P: matrix[Double] = matrix();
var Q: matrix[Double] = matrix();
var pq: matrix[Double] = matrix();
var E: matrix[Double] = matrix();

var steps: Int = 0;
while ( steps < num_steps ) {
  steps += 1;
  for i = 0, n-1 do
      for j = 0, m-1 do {
          pq[i,j] := 0.0;
          for k = 0, d-1 do
              pq[i,j] += P[i,k]*Q[k,j];
          E[i,j] := R[i,j]-pq[i,j];
          for k = 0, d-1 do {
              P[i,k] := P[i,k] ^ (2*a*E[i,j]*Q[k,j]);
              Q[k,j] := Q[k,j] ^ (2*a*E[i,j]*P[i,k]);
          };
      };
};
\end{lstlisting}

\end{document}